\documentclass[11pt,letterpaper]{article}
\usepackage{amsmath,amsthm,amssymb}
\usepackage{graphicx} 
\usepackage{enumitem}
\usepackage{bbm}
\usepackage{algorithmic}
\usepackage{xcolor}
\usepackage{thm-restate}
\usepackage[ruled,vlined,linesnumbered]{algorithm2e}

\usepackage[margin=1in]{geometry}
\usepackage[colorlinks, citecolor = blue, linkcolor = magenta]{hyperref}
\usepackage{cleveref}
\usepackage{tikz-cd}
\usepackage{float}

\numberwithin{equation}{section}

\newtheorem{theorem}{Theorem}[section]
\newtheorem{claim}[theorem]{Claim}

\newtheorem{proposition}[theorem]{Proposition}
\newtheorem{corollary}[theorem]{Corollary}
\newtheorem{lemma}[theorem]{Lemma}

\theoremstyle{definition}

\newtheorem{definition}[theorem]{Definition}

\newtheorem{remark}{Remark}

\newcommand{\norm}[1]{\left\Vert#1\right\Vert}

 \newcommand{\eps}{\varepsilon}
 
\renewcommand{\mid}{\;\middle\vert\;} 
 
 \renewcommand{\d}{\,\-d}

\newcommand{\F}{\mathbb{F}}

\newcommand{\calD}{\mathcal{D}}

\newcommand{\calC}{\mathcal{C}}
\newcommand{\calM}{\mathcal{M}}

\newcommand{\yes}{\mathrm{yes}}
\newcommand{\no}{\mathrm{no}}

\newcommand{\id}{\mathfrak{id}}
\newcommand{\fraki}{\mathfrak{i}}

\renewcommand{\d}{\,\mathrm{d}}

\newcommand{\reff}{\mathrm{ref}}
\newcommand{\supp}{\mathrm{supp}}
\newcommand{\lspan}{\mathrm{span}}
\newcommand{\Dno}{\mathcal{D}_{\mathrm{no}}}
\newcommand{\Dyes}{\mathcal{D}_{\mathrm{yes}}}
\newcommand{\adv}{\mathrm{adv}}

\renewcommand{\Pr}{\operatorname*{\mathbf{Pr}}}
\DeclareMathOperator*{\E}{\mathbf{E}}

\newcommand{\pr}[2][]{ \ifthenelse{\isempty{#1}}
  {\Pr\left[#2\right]} {\Pr_{#1}\left[#2\right]} }
\newcommand{\ex}[2][]{ \ifthenelse{\isempty{#1}}
  {\E\left[#2\right]}
  {\E_{#1}\left[#2\right]} }

\newcommand\skipi{{\vskip 10pt}}
\renewcommand{\leq}{\leqslant}

\renewcommand{\geq}{\geqslant}

\title{Multi-Pass Streaming Lower Bounds for Approximating Max-Cut}
\author{Yumou Fei\thanks{Department of EECS, Massachusetts Institute of Technology.}\and Dor Minzer\thanks{Department of Mathematics, Massachusetts Institute of Technology. Supported by NSF CCF award 2227876 and NSF CAREER award 2239160.}\and Shuo Wang\thanks{Department of Mathematics, Massachusetts Institute of Technology. Supported by NSF CCF award 2227876.}}
\date{\vspace{-5ex}}

\begin{document}
\maketitle
\begin{abstract}
    In the Max-Cut problem in the streaming model, an algorithm is given the edges of an unknown graph $G = (V,E)$ in some fixed order, and its
    goal is to approximate the size of the largest cut in $G$. Improving upon an earlier result of Kapralov, Khanna and Sudan, it was shown by Kapralov and Krachun that for all $\eps>0$, no $o(n)$ memory streaming algorithm can achieve a $(1/2+\eps)$-approximation for Max-Cut. Their result holds for
    single-pass streams, i.e.~the setting in which the algorithm only views the stream once, and it was open whether multi-pass access
    may help. The state-of-the-art result along these lines, due to Assadi and N, rules out arbitrarily good approximation algorithms with constantly many passes and $n^{1-\delta}$ space for any $\delta>0$.

    We improve upon this state-of-the-art result, showing that any non-trivial approximation algorithm for Max-Cut requires either polynomially many passes or polynomially large space. More specifically, we show that for all $\eps>0$, a $k$-pass streaming $(1/2+\eps)$-approximation algorithm for Max-Cut requires $\Omega_{\eps}\left(n^{1/3}/k\right)$ space.
    This result leads to a similar lower bound for the Maximum Directed Cut problem, showing the near optimality of the algorithm of [Saxena, Singer,  Sudan, Velusamy, SODA 2025].

    Our lower bounds proceed by showing a communication complexity lower 
    bound for the Distributional Implicit Hidden Partition (DIHP) Problem, introduced by Kapralov and Krachun. While a naive application of the discrepancy
    method fails, we identify a property of protocols called ``globalness'', and show that (1) any protocol for DIHP can be turned into a global protocol, (2) the discrepancy of a global protocol must be small. The second step 
    is the more technically involved step in the argument, and therein 
    we use global hypercontractive inequalities, and more specifically strong quantitative versions of the level-$d$ inequality for global functions.    
\end{abstract}

\thispagestyle{empty}
\setcounter{page}{0}
\newpage
{
\setcounter{tocdepth}{3} 
\tableofcontents
}
\thispagestyle{empty}
\setcounter{page}{0}
\newpage
\setcounter{page}{1}

\section{Introduction}
The approximability of constraint satisfaction problems (CSPs in short) 
has been a central point of study in complexity theory. By now, thanks 
to the theories of NP-hardness and of probabilistically checkable proofs, we have a fairly good understanding (though incomplete) of approximation ratios achievable by polynomial time algorithms. For example,
for the Max-Cut problem\textemdash which is the focus of this paper\textemdash the best known approximation algorithm due to Goemans-Williamson~\cite{GW95} achieves approximation ratio of $\alpha_{GW}\approx 0.878$, and furthermore this is best possible assuming the Unique-Games Conjecture~\cite{Khot02,KKMO07}.

In this paper we focus on the streaming model~\cite{AMS96},
in which an input CSP instance is given to the algorithm as a stream. Namely, suppose that the instance is $\Psi = (X,E)$ where $X$ is a set of variables, and $E = \{C_1,\ldots,C_m\}$ is a set of constraints. Then the stream consists 
of the constraints of $\Psi$ given in some predetermined order, and the goal of the 
algorithm is to output a number that approximates the 
\emph{value} of $\Psi$, 
which is defined as the largest number of constraints that can be simultaneously satisfied. Here, we think of the algorithm as space bounded (but otherwise it can carry out very costly computations), and the question is how much space is needed/ suffices to approximate CSPs of interest. 
Casting the Max-Cut problem in this lens, we think of the variables as nodes $X = \{x_1,\ldots,x_n\}$, and of the constraints as corresponding to edges, specifying we want their endpoints to be on different sides of the cut. In other words, for each edge $e = \{i,j\}$ in the graph, we have the constraints $x_i+x_j = 1\pmod{2}$. The value of an instance is the maximum number of equations satisfied over all $\mathbb{F}_2$-labeling of the variables.

The study of CSPs in the streaming model has received much attention over the last decade, see for example~\cite{KKS15,guruswami2017streaming,KK19,chou2020optimal,AKSY20,chou2022linear,AN21,chou2022linear,saxena2023streaming,saxena2023improved,hwang2024oblivious,chou2024sketching,saxena2025streaming} (and~\cite{SudanSurvey} for a survey). The starting point of this line of study is the Max-Cut problem, for which it was observed that one can get a  $1/2$-approximation streaming algorithm with $O(\log n)$ space. Indeed, the algorithm simply counts the number of edges 
in the graph and divides it by $2$. Building on the work~\cite{gavinsky2007exponential}, 
the work of~\cite{KKS15} was the first to show that beating this trivial $1/2$-approximation ratio requires a large amount of space. More specifically, they showed 
that for all $\eps>0$, any $(1/2+\eps)$-approximation streaming algorithm requires $\Omega_{\eps}(\sqrt{n})$ memory. This result was later improved in~\cite{KK19}, who
showed a lower bound of $\Omega_{\eps}(n)$ for the same problem, which is tight up to 
logarithmic factors.

Subsequent works have further developed both the algorithmic as well as the lower bound machinery, leading to non-trivial algorithms as well as hardness results for different CSPs. In a sense, the 
main goal is to determine, for each CSP, what is the best approximation ratio achievable 
by an efficient streaming algorithm. Here and throughout, an efficient streaming algorithm
is one that uses poly-logarithmic space. For example, the work~\cite{chou2024sketching} considers a special sub-class of (single-pass) streaming algorithms known as sketching algorithms, and shows that they exhibit a dichotomy behavior. 
Namely, for each CSP there is an approximation ratio achievable by an efficient sketching algorithm, 
and doing any better than 
that (with a sketching algorithm) requires polynomial memory. 
The main open problem thus is to establish a similar dichotomy
behavior for the more general class of streaming algorithms.
While this challenge remains largely open, much progress
has been made over the last decade.

\subsection{Multi-pass Streaming Algorithms}
With a few exceptions that we soon mention, all of the above results are concerned with 
single-pass, worst case order streaming algorithm. By that, we mean that the algorithm views the input stream once, and furthermore that its order is predetermined by an adversary. This naturally leads one to consider other settings of 
streaming algorithms:
\begin{enumerate}
    \item Random order streaming model: what if the input stream is given in a random, uniformly chosen order?
    \item Multi-pass streaming algorithms: what if the algorithm is given mulitple passes on the stream, say $k$ passes where $k>1$?
\end{enumerate}
Indeed, these directions have been highlighted
as two of the most interesting directions in which not much 
is known, see for example~\cite{SudanSurvey}. 
The result of~\cite{KKS15}, showing that any $1/2+\eps$ approximation streaming algorithm for Max-Cut must use $\Omega_{\eps}(\sqrt{n})$ space, also holds in the first model above.
The focus of the current paper is on the second model, which we refer to as the 
multi-pass model. With this in mind, the results of~\cite{KKS15,KK19} only hold in the 
single-pass model. In fact, this is true for almost all of the known lower bounds for streaming algorithms.

Exceptions to this assertion are the works~\cite{AKSY20,AN21,CKP+23}, which 
address the multi-pass setting, and they do so using information theoretic techniques. In particular, 
the works~\cite{AKSY20,AN21} rule out arbitrarily good approximations for 
Max-Cut using sub-linear space and constantly many passes. Their proofs proceed by studying the 
associated cycle detection problem/ cycle counting problem, and showing a communication
complexity lower bound for them. In the cycle counting problem, each one of $K$-players is given some of the edges of an unknown graph, and their goal is to determine whether 
the graph has at least $r$ cycles, or less than $r$ cycles. 
These works prove communication complexity lower bounds for this
problem, and use it (via standard reductions) to show hardness of approximation results for Max-Cut in the multi-pass streaming model. We refer the reader to~\cite{Assadi} for a more detailed discussion on the subject.

The works~\cite{AKSY20,AN21} establish 
$\Omega(n^{1-O(k\eps)})$ lower bound for $(1-\eps)$-approximation of Max-Cut, where $k$ is the number of passes, namely a near linear space
lower bound provided that 
the number of passes is much smaller than $1/\eps$. While giving evidence towards the hardness of Max-Cut in the multi-pass setting, this result leaves open the possibility of, say, a polylog-space $0.9$-approximation using even two passes.

The work~\cite{CKP+23}
studies the search version of the associated cycle finding problem, and rules out $o(\log n)$-pass, $n^{o(1)}$-space algorithms for it. To 
the best of our knowledge, 
it has no direct implications on the approximability of Max-Cut in the multi-pass model.

\subsection{Main Result}
Our main result is a new space lower bound for multi-pass streaming algorithms for Max-Cut.
Our lower bound works against any non-trivial $(1/2+\eps)$-approximations, and remains meaningful for even polynomially many passes on the stream.
\begin{restatable}{theorem}{mainintro}\label{thm:main_into}
    For all $\eps>0$ and $k\in\mathbb{N}$, any randomized, $k$-pass, $\left(1/2+\eps\right)$-approximation streaming algorithm for Max-Cut requires $\Omega_{\eps}(n^{1/3}/k)$ space. 
\end{restatable}
In words,~\Cref{thm:main_into} asserts that even if an algorithm takes, say, $n^{1/6}$ many passes on the stream, 
it still requires polynomial space to achieve a non-trivial approximation for Max-Cut.

\paragraph{Maximum Directed Cut:} Theorem~\ref{thm:main_into} has an interesting corollary for the directed version of Max-Cut. In the Maximum Directed Cut problem one is given a directed graph $G = (V,E)$, and the
goal is to find a partition $V = L\cup R$ of $V$ maximizing the
fraction of edges going from $L$ to $R$. For this problem, the
naive algorithm that counts the number of edges in the graph
gives approximation ratio of $1/4$. 
The work~\cite{chou2020optimal} gives a single-pass, logarithmic space algorithm achieving approximation ratio $4/9$, and furthermore shows that improving this ratio 
requires $\Omega(\sqrt{n})$ space. It turns out that multi-pass algorithms 
can achieve strictly better approximation ratios for this problem, and the work~\cite{saxena2025streaming} gives a logarithmic space, $O(1/\eps)$-pass algorithm with approximation ratio $1/2 - \eps$. 
The following result, which is a direct corollary of Theorem~\ref{thm:main_into}, shows that going beyond $1/2$-approximation requires polynomial space.
\begin{theorem}
    For all $\eps>0$ and $k\in\mathbb{N}$, any randomized $k$-pass $(1/2+\eps)$-approximation streaming algorithm 
    for Maximum Directed Cut requires $\Omega_{\eps}(n^{1/3}/k)$
    space.
\end{theorem}

\paragraph{Proof technique:} The technique in the proof of~\Cref{thm:main_into} is fairly general, and it may apply to wider classes of predicates. 
Indeed, in contrast to the arguments in~\cite{AKSY20,AN21,CKP+23} which are all information 
theoretic and rely fairly heavily on the cycle detection problem and its connection to
the Max-Cut problem, our proof is Fourier analytic, closer in spirit to the arguments
in~\cite{KKS15,KK19}. We defer a detailed discussion of our techniques to~\Cref{sec:pf_techniuqe}.
While we have put some effort into optimizing the quantitative aspects of~\Cref{thm:main_into}, it seems 
that currently our techniques are unable
to establish linear space lower bounds for Max-Cut, and we think that this is an interesting open challenge.

\subsection{Streaming Lower Bounds via Communication Complexity}\label{sec:pf_techniuqe}
As was done in prior works, our proof of~\Cref{thm:main_into} proceeds by establishing a
communication complexity lower bound for a related problem. 
More specifically, the problem we study is the distributional implicit hidden partition problem (abbreviated as DIHP henceforth) from~\cite{KK19}, 
which we define next.
\begin{definition}
Fix a ground set $U$ and an integer $m\leq |U|/2$. 
An element $y\in \{-1,0,1\}^{\binom{U}{2}}$ is said to be a \emph{labeled matching} over $U$ if the edge set $\supp(y):=\left\{\{u,v\}\in\binom{U}{2}:y_{\{u,v\}}\neq 0\right\}$ consists of vertex disjoint edges. In that case, we think of the support of $y$ as a graph, and of the label of an edge $\{u,v\}$ as $y_{\{u,v\}}$. 
\end{definition}
\begin{definition}\label{def:matchings}
    The space of labeled matchings, denoted by $\Omega^{U,m}\subseteq \{-1,0,1\}^{\binom{U}{2}}$, is defined as
\[
\Omega^{U,m}:=\left\{y\in \{-1,0,1\}^{\binom{U}{2}}:\supp(y)\text{ is a matching with }m\text{ edges}\right\}.
\]
\end{definition}

Throughout this paper we will mostly consider the case the ground set $U$ is $[n]$, and the matching size $m$ is $\alpha n$ where $\alpha>0$ is a small absolute constant. When both the ground set and the matching size are clear from context, we often abbreviate notations and write $\Omega$ instead of $\Omega^{[n],\alpha n}$. In particular, we denote by $\Omega^K$ the Cartesian product of $K$ copies of $\Omega$.

With this in mind, in the DIHP problem we have $K$ players,
each receiving a labeled matching from 
$\Omega^{[n],\alpha n}$ as an input. Their goal is to be able to tell
if the labels of their matchings are consistent, in
the sense that there is a bipartition of the vertex set 
$U$ so that edges that cross this bipartition are labeled by 
$-1$, and edges that stay within one side of the bipartition
are labeled by $1$. We stress here that each player only gets
to see the edges they received, so while that player can find
bipartitions of the vertices consistent with their edges, 
the challenge here is that the bipartition should be consistent
with the edges of other players as well.

\paragraph{Lower bounds for DIHP:}
To prove that this communication problem is hard,~\cite{KK19} introduce two distributions over the inputs of the players. In the YES distribution, all of the labels follow from 
one common bipartition of the vertices, whereas in the NO distribution each matching is labeled independently.
Formally, we define these distributions as follows:
\begin{definition}
Let $K\in\mathbb{N}$ be a constant. Define two distributions $\calD_{\yes}$ and $\calD_{\no}$ over $\Omega^{K}$:
\begin{enumerate}
\item {\bf The no distribution:} define $\calD_{\no}$ to be the uniform distribution over $\Omega^{K}$. 
\item {\bf The yes distribution:} 
sample a uniformly random vector $x\in\mathbb{F}_{2}^{n}$, then independently and uniformly sample $K$ matchings $M^{(1)},M^{(2)},\dots,M^{(K)}$ of size $\alpha n$. For each $i\in[K]$, we let $y^{(i)}\in\Omega$ have support $\supp(y^{(i)})=M^{(i)}$ and be defined as $y_{uv}^{(i)}=(-1)^{x_{u}+x_{v}}$ for $\{u,v\}\in M^{(i)}$. We define $\calD_{\yes}$ to be the joint distribution of $(y^{(1)},\dots,y^{(K)})$ obtained by this procedure.
\end{enumerate}
\end{definition}
The work~\cite{KK19} shows an $\Omega(n)$ communication 
lower bound for this problem for protocols in which player $1$
sends a message to player $2$, which sends a message to player $3$ and so on until its player $K$ turn to speak. In that step, player $K$ should decide whether to accept or reject. These type of
restricted protocols suffice for establishing single-pass lower bounds, and are 
typically easier to prove.

A subtle difference between the communication problem 
we consider here and the one used in \cite{KK19}, is that 
in our case each player only knows their matching. In contrast, in the setting of~\cite{KK19}, player $i$ knows all of the 
matchings $M^{(1)},\ldots,M^{(i)}$, and only the labels are
private. 
This does not matter for~\cite{KK19}, as their setting is specifically designed to facilitate their single-pass analysis. 
This distinction is crucial in the multi-pass setting, though, as without it the problem 
is no longer hard. The multi-pass
setting corresponds to 
protocols similar to the 
above, except that in the
end, player $K$ sends a message to player $1$, and 
then the protocol continues in the same way.
Note that with high probability the matchings $M^{(1)},\dots, M^{(K)}$ will not be edge-disjoint. Hence, if the matchings were public, the last player to act could announce a common edge to two of the matchings, and in the second pass the two corresponding players could broadcast their labels of that edge and compare it. In the no distribution, with probability $1/2$ the labels would not match each other. 

In light of this, we will consider the version of this problem wherein both the matchings and the labels are private.
\begin{definition}[Distributional Implicit Hidden Partition  Problem]\label{def:DIHP}
We define $\mathsf{DIHP}(n,\alpha, K)$ to be the $K$-player communication game in the number-in-hand (NIH) communication model where the $i$-th player gets as private input $y^{(i)}\in\Omega$, and their goal is to decide whether $(y^{(1)},\dots,y^{(K)})$ comes from $\calD_{\yes}$ or $\calD_{\no}$. For a protocol $\Pi$, we define its ``advantage'' to be
\[
\mathrm{adv}(\Pi):=\left|\Pr_{y^{(1)},\dots,y^{(K)}\sim \calD_{\yes}}\left[\Pi\left(y^{(1)},\dots,y^{(K)}\right)=1\right]-\Pr_{y^{(1)},\dots,y^{(K)}\sim \calD_{\no}}\left[\Pi\left(y^{(1)},\dots,y^{(K)}\right)=1\right]\right|.
\]
\end{definition} 

Our main result regarding the $\mathsf{DIHP}(n,\alpha, K)$ 
is the following communication complexity lower bound:
\begin{restatable}{theorem}{thmmain}\label{thm:main}
    Let $\alpha\in (0,10^{-7}]$ be a constant. Any communication protocol $\Pi$ for $\mathsf{DIHP}(n, \alpha, K)$ with advantage at least $0.1$ requires ${\Omega}\left(n^{1/3}K^{-2}\right)$ bits of communication\footnote{The communication complexity is the total number of bits broadcast by all players in all rounds during the communication game.}.
\end{restatable}
\Cref{thm:main_into} follows quickly from~\Cref{thm:main}; see~\Cref{sec:apx_quick_pf} for a formal argument. 
Our proof of~\Cref{thm:main} combines a type of ``structure-vs-randomness'' argument along with a global hypercontractive inequality. The structure-vs-randomness approach has achieved great success in proving lifting theorems and lower bounds for search problems in the area of communication complexity, see for example~\cite{RM97,GPW17,YZ24}. In our case, it refers to the fact
that we are able to take any protocol for~\Cref{thm:main} 
and break it into a ``structured'' part, in which players 
fully expose information about edges, and a ``pseudo-random'' part, 
in which players do not expose too much information about
any small set of individual edges. We refer to the last property as ``globalness''
of protocols, taking inspiration from the notion of 
globalness for sets/ functions in the context of discrete Fourier analysis. 

To analyze global protocols, we use 
global hypercontractive inequalities, a tool from discrete Fourier analysis that has its origins in the theory of PCPs~\cite{KMS1,DKKMS1,DKKMS2,KMS} and by now has found applications throughout discrete mathematics~\cite{LM,EKLapp,KM,keevash2023forbidden,MZoptimal,MZrmgen,KLM23,keevash2024largest,LMgp,KLLM24,green2024improved}. Most relevant to us is the result of
Keller, Lifshitz and Marcus~\cite{KLM23} which establishes
quantitatively optimal global hypercontractive inequalities 
for product spaces. At a high level, this is because 
we are concerned with functions over the space
$\Omega^{[n],\alpha n}$ as per~\Cref{def:matchings}, which is close in spirit to a product distribution over $\{-1,0,1\}^{\binom{[n]}{2}}$.

In the next section, we provide a more detailed discussion of our techniques and an overview of the proof.

\subsection{Proof Overview}\label{sec:proof_overview}
\subsubsection{The Discrepancy Method via Globalness}
As is often the case in communication complexity, 
we prove~\Cref{thm:main} using the discrepancy method, which we state below for convenience.  
\begin{lemma}\label{lem:disc_easy}
 Let $\Pi$ be a deterministic 
 protocol for $\mathsf{DIHP}(n,\alpha, K)$, 
 and let $\mathcal{R}$ be the partition it induces over 
 of $\Omega^K$ into combinatorial rectangles. Then
 \[
 \mathrm{adv}(\Pi) \leq \frac{1}{2}\sum_{R\in\mathcal{R}}|\mathcal{D}_{\yes}(R) - \mathcal{D}_{\no}(R)|. 
 \]
\end{lemma}
\begin{proof}
The proof follows from definitions quickly, after breaking down each 
one of the probabilities in the definition of $\text{adv}(\Pi)$ according
to the rectangle $R\in\mathcal{R}$ the input lies in.
\end{proof}

\paragraph{A naive application of the discrepancy method:} 
by~\Cref{lem:disc_easy}, our result would follow if we were able to show that for any 
sizable rectangle $R$, we have that the probability masses $\mathcal{D}_{\text{yes}}(R)$ and $\mathcal{D}_{\text{no}}(R)$ are 
very close to each other. Indeed, if we were able to show that for any rectangle $R$ such that
$\mathcal{D}_{\mathrm{no}} (R)\geq 2^{-n^c}$ (where $c$ is a small constant), it holds that
\begin{equation}\label{eq:1_pf_ov}
    |\mathcal{D}_{\mathrm{yes}}(R) - \mathcal{D}_{\mathrm{no}}(R)|\leq 0.01 \cdot 
\mathcal{D}_{\mathrm{no}}(R),
\end{equation} 
then we would get (by summing up) that the advantage of $\Pi$ is at most $0.01 + |\mathcal{R}|2^{-n^c}$, which quickly yields 
a $\Omega(n^c)$ communication
lower bound. Alas, there are sizable combinatorial rectangles $R$ for which~\eqref{eq:1_pf_ov} fails. Indeed, consider the rectangle 
$R = A^{(1)}\times\ldots\times A^{(K)}$ where $A^{(i)}$ is the collection
of labeled matchings containing the edge $e = \{1,2\}$ with label 
$(-1)^i$. Then we have that $\mathcal{D}_{\mathrm{no}} (R)\approx (\alpha/n)^{K}$ meaning that $R$ is sizable, but $\mathcal{D}_{\mathrm{yes}}(R) = 0$ so the gap between them is almost as big as it can be. 

Thus, we see 
that there are in fact rectangles violating~\eqref{eq:1_pf_ov}. The
example we have seen though is very special, and it heavily relies on knowing
full information about one specific edge $e$. Given that a sizable rectangle cannot give a lot of information about many specific edges, 
there does not seem to be a way to use such rectangles to construct a protocol for DIHP. Can we prove that these are
the only violations of~\eqref{eq:1_pf_ov}? 
If true, can we use such a
result to prove a lower 
bound for $\mathsf{DIHP}(n,\alpha, K)$ via (an appropriate adaptation of) the discrepancy method?

\paragraph{Global rectangles:} 
the observation above motivates us to define a certain type of sets called the \textit{global sets}. Informally speaking, a set
$A\subseteq \Omega$ is called
global if no edge $e$ is significantly more prevalent than others in members in $A$; we refer the reader to~\Cref{def:global_set} for a more precise definition.
This notion naturally extend to rectangles: $R=A^{(1)}\times\cdots\times A^{(K)}$ is called a \textit{global rectangle} if for every $i$, $A^{(i)}$ is a global set. 

With the notion of global rectangles, one may wonder if indeed the only
rectangles violating~\eqref{eq:1_pf_ov} are not global. In other words, 
one may wonder if~\eqref{eq:1_pf_ov} holds for for every global rectangle $R$ with $\mathcal{D}_{\mathrm{no}}(R)\geq 2^{-O(n^c)}$. We show that this is
indeed true for $c = 1/3$; see~\Cref{lem:main} for a precise statement.
The proof of this uses discrete Fourier analysis, and in particular global hypercontractivity.

\paragraph{Global protocols:} having shown that~\eqref{eq:1_pf_ov} holds for global rectangles, one may wonder how to convert it into a communication complexity lower bound. Clearly, an arbitrary protocol
$\Pi$ need not induce a partition of $\Omega^{K}$ into global rectangles.
However, we prove a sort of regularity lemma, showing that $\Pi$ 
can be modified to obtain this property while roughly keeping the same
communication complexity and advantage. We refer to such protocols as global protocols. At a high level, if a protocol $\Pi$
is not global and we take a rectangle $R = A^{(1)}\times\ldots\times A^{(K)}$ in its induced partition that is not global, and consider any $i$ and a labeled edge $(e,b)$ such that $(e,b)$ is significantly more likely to appear in $A^{(i)}$. We then modify the protocol by saying that
if the players reached rectangle $R$, then player $i$ exposes the bit corresponding to $(e,b)$. Since each time a player 
exposes a bit due to lack of globalness the relative density of $A^{(i)}$ increases by at least a constant factor, player $i$ will overall expose at most $\log(|\Omega|/|A^{(i)}|)$ many edges
before getting to the case their set is global. Thus, overall
this operation leads to a protocol in which players expose 
on average only a few more edges compared to the original 
protocol (roughly $O(n^{1/3})$ in our case), and besides that the rectangles produced in the 
rectangle partition of the protocol are all global.

\paragraph{Finishing the proof:} given a 
protocol $\Pi$ that solves $\mathsf{DIHP}(n,\alpha, K)$, we first transform
it via our regularity lemma to a global protocol $\Pi^{\reff}$ that roughly has the same advantage and communication complexity.  We then prove an variant of the assertion that~\eqref{eq:1_pf_ov} holds for global rectangles $R$, tailored to the case some edges are fixed. 
Indeed, in our setting we have to deal with rectangles $R = A^{(1)}\times\ldots\times A^{(K)}$ where each $A^{(i)}$ lives inside a 
(possibly different) restriction of $\Omega$. 
This setting is slightly more complicated compared to the (cleaner setting) of global rectangles $R\subseteq \Omega^{K}$, but we show that~\eqref{eq:1_pf_ov} still works. Morally speaking, combined with the fact that a protocol where each player
exposes $O(n^{1/3})$ of their edges fails, this shows that the advantage of $\Pi^{\reff}$ (and hence of $\Pi$) must be small.

\subsubsection{Applying Global Hypercontractivity}
We now give some intuition 
as to the relevance of global hypercontractivity to us. 
Consider a global rectangle $R = A^{(1)}\times\ldots\times A^{(K)}$, fix some $i\in\{1,\ldots,K\}$, and let $A = A^{(i)}$. By reaching the rectangle $R$, player $i$ has communicated the fact that their input belongs to the set $A$. It is natural to ask what sort of partitions $x\in \{0,1\}^n$ of the vertex set are consistent with that behavior. Indeed, if for $i\neq i'$ the partitions consistent with $A^{(i)}$ are very different from the partitions that are consistent with $A^{(i')}$, then the  players may (rightfully) suspect that the input has been drawn from $\mathcal{D}_{\text{no}}$.

This line of thought leads to considering the induced distribution of the partition $x$ conditioned on $A$. More precisely, suppose that we sample a bipartition $x\in\{0,1\}^n$ and a labeled matching $y\in \Omega$ consistent with $x$. What 
is the distribution of $x$ conditioned on $y\in A$? For any fixed $y$, the set of bipartitions $x$ that are consistent with it is the affine subspace
\[
L(y) = \left\{x\in\mathbb{F}_2^n~|~\forall\{u,v\}\in \supp(y), 
~y_{\{u,v\}} = (-1)^{x_u+x_v}
\right\}.
\]
Thus, the distribution of $x$ conditioned on $y$ is the uniform convex combination of the uniform distributions over $L(y)$, over all $y\in A$. Intuitively, the fact that the set $A$ is global means that for typical $y,y'\in A$, the subspaces $L(y)$ and $L(y')$ look ``independent'', and we use this to show that the distribution of $x$ is nearly uniform. A good model case to keep in mind is the case $A$ is a randomly chosen set of some specified density, in which case this can be proved via standard concentration inequalities.
While our actual case wherein $A$ is global cannot be analyzed this way, intuition suggests 
that as no edge is too prevalent in elements in $A$, the intersection of $\supp(y)$, $\supp(y')$ for typical $y,y'\in A$ behaves like in the random case, and that this determines the distribution of $x$.

To formalize this argument, we consider the Fourier expansion of the function $x\mapsto \Pr[x|y\in A]$, and relate 
its Fourier coefficients to Fourier coefficients of 
$1_A\colon \Omega \to\{0,1\}$. The above assertions (morally) translate to proving that 
the Fourier mass on small degrees is small, and this
is where global hypercontractivity and level $d$ inequalities for global functions enters the picture.

\section{Transforming Protocols into Global Protocols}\label{sec:proof_lower_bound}
The goal of this section is to formally define the notions of global sets, rectangles and protocols. We then prove that any protocol for DIHP can be transformed into a global protocol at only a mild cost in its communication complexity. 
Finally, we formally state the discrepancy bound for 
global protocols, and prove it modulo a discrepancy bound for global rectangles which is proved in subsequent sections.
\subsection{Pseudorandomness Notions}
Throughout this section, $\Omega=\Omega^{[n],\alpha n}$ will be a space of labeled matchings as per~\Cref{def:matchings}, and $\mu$ will be the uniform measure over $\Omega$. Thus, for a subset $A\subseteq \Omega$ we have that $\mu(A)$ is the relative size of $A$ inside $\Omega$. 
To define the notion of globalness we shall be interested in
restrictions of $\Omega$ and $\mu$, by which we mean the resulting probability space once we fix a few of the coordinates. 
\begin{definition}[Restrictions]\label{def:restrictions}
For each string $z\in\{-1,0,1\}^{\binom{[n]}{2}}$, we let $\Omega_{z}\subseteq \Omega$ be the restricted domain defined by
\[
\Omega_{z}:=\left\{y\in \Omega:y_{uv}=z_{uv}\text{ for all }\{u,v\}\in\supp(z)\right\}.
\]
We call $z$ a ``restriction'' if $\Omega_{z}\neq \emptyset$. 
In particular, for $z$ to qualify as a restriction, the edge set $\supp(z)$ must be a matching of size no more than $\alpha n$.
\end{definition}

Note that in the space $\Omega$, if we 
fix the entry corresponding to edge $\{i,j\}\in\binom{[n]}{2}$ 
to a non-zero value, then the entry corresponding to any other 
edge $\{i,j'\}$ and $\{i',j\}$ must be given the value $0$. 
Indeed, this is because each element in $\Omega$ is a signed 
matching. This motivates the following definition:
\begin{definition}
For a string $z\in \{-1,0,1\}^{\binom{[n]}{2}}$, we denote by $N(z)\subseteq [n]$ the set of vertices incident to some pair in $\supp(z)$.
\end{definition}

Writing $\Omega = \Omega^{[n],\alpha n}$, note that the restricted 
space $\Omega_z$ is ``isomorphic'' to $\Omega^{[n]\setminus N(z),\, \alpha n-|\supp(z)|}$. 

We will often want to further restrict an already restricted space. Towards this end, for a restriction $z'$ we wish to 
consider the restrictions $z$ that extend/ subsume it:
\begin{definition}
For two strings $z,z'\in \{-1,0,1\}^{\binom{[n]}{2}}$, we say $z$ subsumes $z'$ if 
$\supp(z')\subseteq \supp(z)$ and for all $\{u,v\}\in\supp(z')$ we have that $z_{uv}=z'_{uv}$.
\end{definition}

Armed with the notion of restrictions, we can now formally define
the notion of global sets.
\begin{definition}\label{def:global_set}
A subset $A\subseteq \Omega$ is said to be $z'$-global if $A\subseteq \Omega_{z'}$, and for all restrictions $z$ that subsume $z'$ we have 
\[
\frac{|A\cap \Omega_{z}|}{|\Omega_{z}|}\leq 2^{|\supp(z)|-|\supp(z')|}\cdot \frac{|A\cap \Omega_{z'}|}{|\Omega_{z'}|}.
\]
When $z' = \vec{0}$ is the trivial restriction, we simply say 
that $A$ is global (omitting the $z'$).
\end{definition}
In words, for a set $A$ and a restriction $z'$, we say that $A$ is 
$z'$-global if any further restrictions $z$ that subsumes $z'$ increases
the relative density of $A$ by factor at most $2^{|\supp(z)|-|\supp(z')|}$. We next define global rectangles. 
\begin{definition}[$\zeta$-global rectangles]
    Let $ R = A^{(1)} \times \cdots \times A^{(K)} \subseteq \Omega^{K} $ be a rectangle, and let 
    $\zeta = (z^{(1)},\ldots,z^{(K)})$ be a sequence of restrictions. 
    We say that $R$ is $\zeta$-global
     if for all $i$'s, $A^{(i)}$ is
    $z^{(i)}$-global. Abusing notations, we define  
    $|\zeta|:=\sum_{i=1}^K \left|\supp\left(z^{(i)}\right)\right|$.
\end{definition}

\subsection{The Decomposition Lemma}
We now show that any subset $S$ with large size can be decomposed into global subsets with large size by restricting (on average) not too many coordinates. More precisely: 
\begin{lemma}[Decomposition lemma]\label{lem:reg}
    Let $z'$ be any restriction, and let $A\subseteq \Omega_{z'}$ be any subset. Then we can decompose $A$ into a disjoint union of subsets $A_{(1)},A_{(2)},\dots,A_{(k)}$ such that:
    \begin{enumerate}
        \item {\bf Globalness:} for each $i\in [k]$, there exists a restriction $z_{(i)}$ that subsumes $z'$ such that $A_{(i)}\subseteq \Omega_{z_{(i)}}$ and $A_{(i)}$ is $z_{(i)}$-global.
        \item {\bf Size of the restrictions}: the restrictions $z_{(i)}$ satisfy the following inequality:
        \[\sum_{i=1}^k \frac{|A_{(i)}|}{|A|} \left(\left|\supp(z_{(i)})\right| +\log_{2} \frac{|\Omega_{z_{(i)}}|}{|A_{(i)}|}\right)\leq \left|\supp(z')\right|+\log_{2}\frac{|\Omega_{z'}|}{|A|}+2.\]
    \end{enumerate}
\end{lemma}

\begin{proof}
The proof of the lemma is algorithmic, and the decomposition 
is the result of the following procedure:
\begin{algorithm}
\DontPrintSemicolon
\SetKwInOut{Input}{Input}\SetKwInOut{Output}{Output}
    \caption{Decompose$(A,z')$}\label{alg:decomposition}
    \Input{a restriction $z'$ and a set $A\subseteq \Omega_{z'}$}
    \Output{a sequence of sets $A_{(1)},\dots,A_{(k)}$ and a sequence of restrictions $z_{(1)},\dots,z_{(k)}$}
    $i\gets 0$\;
    \While {$A$ is not $z'$-global}{
        $i \gets i+1$\;
        find a restriction $z_{(i)}$ with largest possible support size such that $z_{(i)}$ subsumes $z'$ and 
        \begin{gather}\label{eq:restriction-maximality}
        \frac{\left|A\cap \Omega_{z_{(i)}}\right|}{\left|\Omega_{z_{(i)}}\right|}> 2^{\left|\supp(z_{(i)})\right|-|\supp(z')|}\cdot\frac{|A|}{|\Omega_{z'}|}
        \end{gather}
        
        $A_{(i)}\gets A\cap\Omega_{z_{(i)}}$\;
        $A\gets A\setminus A_{(i)}$\; 
    }
    \If{$A$ is nonempty (and $z'$-global)}{
        $A_{(i+1)} \gets A$\; 
        $z_{(i+1)}\gets z'$\;
    }
\end{algorithm}

To check the first property, assume on the contrary that some $A_{(i)}$ is not $z_{(i)}$-global. Then by \Cref{def:global_set}, there must exists a restriction $z$ that subsumes $z_{(i)}$ and 
\begin{equation}\label{eq:assume-not-global}
\frac{|A_{(i)}\cap \Omega_{z}|}{|\Omega_{z}|}>2^{|\supp(z)|-\left|\supp(z_{(i)})\right|}\cdot\frac{\left|A_{(i)}\cap \Omega_{z_{(i)}}\right|}{\left|\Omega_{z_{(i)}}\right|}.
\end{equation}
Since the choice of $z_{(i)}$ satisfies \eqref{eq:restriction-maximality} and $A_{(i)}=A\cap \Omega_{z_{(i)}}$, we can combine \eqref{eq:assume-not-global} and \eqref{eq:restriction-maximality} and get
\[
\frac{|A\cap \Omega_{z}|}{|\Omega_{z}|}>2^{|\supp(z)|-\left|\supp(z')\right|}\cdot\frac{\left|A\cap \Omega_{z'}\right|}{\left|\Omega_{z'}\right|}.
\]
This contradicts the maximality of the choice of $z_{(i)}$ since $|\supp(z)|>\left|\supp(z_{(i)})\right|$.

We now check the second property, and towards that end we denote $A_{(\geq i)} = \bigcup_{j=i}^k A_{(j)}$. 
Taking logs of~\eqref{eq:restriction-maximality} gives 
\[
\log_{2} \frac{|\Omega_{z_{(i)}}|}{|A_{(i)}|} + \left| \supp(z_{(i)})\right| \leq \log_{2} \frac{|\Omega_{z'}|}{|A_{(\geq i)}|} + |\supp(z')| = \log_{2} \frac{|\Omega_{z'}|}{|A|} + \log_{2} \frac{|A|}{|A_{(\geq i)}|}+ |\supp(z')|. 
\]
Multiplying this inequality by $|A_{(i)}|/|A|$ and summing over all $i\in [k]$ we get: 
\begin{align*}
    \sum_{i=1}^k \frac{|A_{(i)}|}{|A|} \left(|\supp(z_{(i)})|+\log_{2} \frac{|\Omega_{z_{(i)}}|}{|A_{(i)}|}\right)&\leq \log_{2}\frac{|\Omega_{z'}|}{|A|}+|\supp(z')|+\sum_{i=1}^k \frac{|A_{(i)}|}{|A|}\cdot \log_{2} \frac{|A|}{|A_{(\geq i)}|}\\ 
    &\leq \log_{2}\frac{|\Omega_{z'}|}{|A|}+|\supp(z')|+\int_0^1 \log_{2} \frac{1}{1-x} \d x\\
    &\leq \log_{2}\frac{|\Omega_{z'}|}{|A|} +|\supp(z')|+2,
\end{align*}
where the second transition is because $\sum_{i=1}^k \frac{|A_{(i)}|}{|A|}\cdot \log_{2} \frac{|A|}{|A_{(\geq i)}|}$ is a lower Riemann sum for
the function $f(x) = \log_{2}\left(1/(1-x)\right)$ with points $x^{(i)} = \frac{|A_{(1)}|}{|A|}+\dots+\frac{|A_{(i-1)}|}{|A|}$, and the last transition is by a direct calculation.
\end{proof}

\subsection{From Arbitrary Protocols to Global Protocols}



We next show how to use~\Cref{lem:reg} to transform any protocol into a global protocol. 
Before doing so, we must formally define what we mean by ``global protocols'' and how we measure their cost.
The most natural way to define a global protocol is by saying that at each
node in the communication tree, the set of inputs leading up to it form a global rectangle. 
This structure seems a bit too strict to achieve however, and instead we consider a slight relaxation which is essentially as good as the natural candidate. This adaptation is slightly more complicated to explain, and involves the notion of rounds of communication. In 
each round of communication only a single player speaks, and the globalness requirement asserts that at the end of each round of communication, the set of inputs leading up to that node is global.

To quantize the information revealed in the communication process, we define the following potential function for rectangles. 
\begin{definition}[Potential function of rectangles]
    For restrictions $\zeta=(z^{(1)},\dots,z^{(K)})$ and a rectangle $R=A^{(1)}\times \cdots\times A^{(K)}$ such that $A^{(i)}\subseteq \Omega_{z^{(i)}}$,  we define the potential of $(\zeta,R)$ as: 
    \begin{align*}
    p(\zeta,R):= \sum_{i=1}^K |\supp(z^{(i)})|+ \log_{2} \left(\frac{|\Omega_{z^{(i)}}|}{|A^{(i)}|}
    \right). 
    \end{align*}
\end{definition}
We now formally define global protocols.
\begin{definition}[Global protocols]\label{def:global_protocol}
    A communication protocol $\Pi$ for $\textsf{DIHP}(n,\alpha,K)$ is called an $r$-round global communication protocol if it specifies the following procedure of communications: 
    \begin{itemize}
        \item the $K$ players take turns to send messages according $\Pi$; 
        \item there are at most $r$ rounds of communications, there is only one player sending message in a single round;
        \item the length of message in each round of communications is not bounded; instead, from the perspective of rectangles, after each round of communications, a $\zeta$-global rectangle $R$ is further partitioned into several global rectangles $R: = R_{(1)}\cup\dots\cup R_{(k)}$ such that: (1) $R_{(i)}$ is $\zeta_{(i)}$-global; (2) $\zeta_{(i)}$ subsumes $\zeta$; (3) the following inequality holds: 
        \begin{align*}
            \sum_{i=1}^k\frac{|R_{(i)}|}{|R|}p(\zeta_{(i)},R_{(i)})\leq p(\zeta,R) + 3. 
        \end{align*}
    \end{itemize}
\end{definition}
We now show an explicit construction of global protocol $\Pi^{\reff}$, given any communication protocol $\Pi$, and show that $\adv(\Pi^{\reff})\geq \adv(\Pi)$. 
\begin{lemma}\label{lem:arbitrary_to_global}
    Given a communication protocol $\Pi$ for $\mathsf{DIHP}(n,\alpha,K)$ with communication complexity at most $r$, we can construct an $r$-round global protocol $\Pi^{\reff}$ for $\mathsf{DIHP}(n,\alpha,K)$ such that $\adv(\Pi^{\reff}) \geq \adv(\Pi)$. 
\end{lemma}
\begin{proof}
     We start with some setup. For convenience, given an arbitrary communication protocol $\Pi$ with $|\Pi|=r$, we consider its tree structure. Without loss of generality, assume that at each round, a player sends exactly one bit of message. In this case, the communication tree is a binary tree. Furthermore, we extend the tree so that all leaf nodes lie at the same depth. In particular, these modifications do not increase the communication cost or decrease the advantage of $\Pi$.  Each node $u$ on the tree has an associated rectangle $R_u = A^{(1)}\times \cdots\times A^{(K)}$. We use $\mathcal{N}_d$ to denote the set of all rectangles (nodes) of $\Pi$ of depth $d$, where root node is of depth $0$. In particular, $\mathcal{N}_{r}$ denotes the set of all leaf rectangles (nodes) of $\Pi$, and each leaf rectangle (node) is labeled with an output, either ``1'' or ``0''. 

    With the setup described above, we now construct the global protocol $\Pi^{\reff}$ (where the superscript ``$\reff$'' stands for ``refined''). 
    The formal construction of $\Pi^{\reff}$ is described in Algorithm \ref{alg:refinement}, but it is helpful to think
    of the construction slightly less formally. Note that viewing the protocol $\Pi$ as a communication tree, we have that each node in it corresponds to a rectangle. Thus, 
    in $\Pi^{\reff}$ we proceed going over the nodes of this tree, starting with the root node of $\Pi$, and decompose each one of these rectangle into global rectangles using~\Cref{lem:reg}. 
\begin{algorithm}
\DontPrintSemicolon
\SetKwInOut{Input}{Input}\SetKwInOut{Output}{Output}
    \caption{Construction of the global protocol $\Pi^{\reff}$}\label{alg:refinement}
    \Input{the $i$-th player gets input $y^{(i)}\in\Omega$}
    \Output{a bit $\texttt{ans}\in\{0,1\}$}
    initialize: $v\leftarrow$ the root of $\Pi$; for every $i\in [K]$, $A^{(i)} \leftarrow \Omega,z^{(i)}  = \emptyset$; $R\leftarrow A^{(1)}\times \cdots\times A^{(K)}$\;
    \While {\text{$v$ is not a leaf node}}{
        suppose player $i$ communicates a bit at node $v$ according to $\Pi$\;
        let $A^{(i)}=A_{0}\cup A_{1}$ be the partition\footnotemark~at $v$ according to $\Pi$\;
        let $b \in\{0,1\}$ be such that $y^{(i)}\in A_{b}$\;
        player $i$ sends $b$, and we update $A^{(i)}\leftarrow A_{b}$, $R\leftarrow A^{(1)}\times \cdots\times A^{(K)}$, $v\gets v_{b}$\;
        \If{$A^{(i)}$ is not $z^{(i)}$-global}{
            $(A_{(1)},z_{(1)}),\dots,(A_{(k)},z_{(k)})\leftarrow\mathrm{Decompose}(A^{(i)},z^{(i)})$ (running \Cref{alg:decomposition})\;
            let $\ell\in [k]$ be such that $y^{(i)}\in A_{(\ell)}$\; 
            player $i$ sends $\ell$, and we update $A^{(i)}\leftarrow A_{_{(\ell)}}, z^{(i)} \leftarrow z_{(\ell)}$, $R\leftarrow A^{(1)}\times \cdots\times A^{(K)}$\;
        }
    }
    let $\texttt{ans} = 1$ if $\mathcal{D}_{\mathrm{yes}}(R)\geq \mathcal{D}_{\mathrm{no}}(R)$, otherwise let $\texttt{ans} = 0$\;
    output $\texttt{ans}$\;
\end{algorithm}
\footnotetext{
Strictly speaking, the protocol $\Pi$ at node $v$ does not directly divide the set $A^{(i)}$ itself, since $A^{(i)}$ is a set dynamically maintained during the execution of the \emph{refined} protocol $\Pi^{\reff}$. However, there always exists a superset $A^{(i)}_{\text{original}}\supseteq A^{(i)}$ that is divided at node $v$ into two subsets based on the message of player $i$ in the original protocol $\Pi$. The partition $A^{(i)}=A_{0}\cup A_{1}$ is then the restriction of the partition of $A^{(i)}_{\text{original}}$ according to $\Pi$. 
}
    
    \paragraph{$\Pi^{\mathrm{ref}}$ is global:}
    to analyze the protocol $\Pi^{\mathrm{ref}}$, we 
    define a round of communication as the event in which a player $i$ sends both a bit $b$ and an integer $\ell$ (see Lines 3–10 in Algorithm \ref{alg:refinement}). Note that every rectangle produced by the refined protocol $\Pi^{\reff}$ after rounds of communications is a global rectangle with some associated restriction. We will keep track of the restriction corresponding to each rectangle $R$. We define $\mathcal{N}^{\mathrm{ref}}_d$ to be the set of all restriction-rectangle pairs $(\zeta, R)$ that are generated by $\Pi^{\mathrm{ref}}$ after the first $d$ rounds of communication. There are two subtle differences between $\mathcal{N}_d$ and $\mathcal{N}^{\mathrm{ref}}_d$:
\begin{enumerate}
    \item $\mathcal{N}_d$ is a set of rectangles, whereas $\mathcal{N}^{\mathrm{ref}}_d$ is a set of pairs, each consisting of a restriction and a rectangle.
    \item The notion of ``depth'' differs: a rectangle $R \in \mathcal{N}_d$ is obtained by protocol $\Pi$ after exactly $d$ bits have been communicated, while a pair $(\zeta, R) \in \mathcal{N}^{\mathrm{ref}}_d$ is produced by protocol $\Pi^{\mathrm{ref}}$ after $d$ rounds of communication, with each round involving the transmission of a bit $b \in \{0,1\}$ and an integer $\ell$.
\end{enumerate}
First, we show that $\Pi^{\reff}$ is a global protocol as per~\Cref{def:global_protocol}. The discussion above shows that (1) $\Pi^{\reff}$ has exactly $r$ rounds of communications; (2) after $0\leq d\leq r$ rounds of communications, the resulting rectangles in $\mathcal{N}_d^{\reff}$ are all global rectangles with restrictions; (3) for all $(\zeta,R)\in \mathcal{N}_{d-1}^{\reff}$ and $(\zeta',R')\in \mathcal{N}_{d}^{\reff}$ such that $R'\subseteq R$, we have $\zeta'$ subsume $\zeta$. Thus, we 
have the first two items in~\Cref{def:global_protocol}, and
we next show the third item.

It suffices to upper bound the potential increment after each round of communications. Assume that after $d$ rounds of communication according to $\Pi^{\reff}$, we obtain a pair $(\zeta,R)\in\mathcal{N}_{d}^{\reff}$ and player $i$ will speak in the next round. The communication of the player $i$ divides $A^{(i)}$ into two parts $A_0,A_1$, which decomposes the rectangle $R$ into two disjoint rectangles $R_0,R_1$ via the message of $b$ (see lines 3 to 6 in Algorithm \ref{alg:refinement}). In lines 7 to 10, $R_0$ and $R_1$ are further decomposed into several global rectangles separately by the message of $\ell$. We have:
\begin{align*}
    \sum_{\substack{(\zeta',R')\in \mathcal{N}_{d+1} ^{\reff}\\ R'\subseteq R}}\frac{|R'|}{|R|}\cdot  p(\zeta',R') &= \frac{|R_0|}{|R|}\sum_{\substack{(\zeta',R')\in \mathcal{N}_{d+1} ^{\reff}\\ R'\subseteq R_0}}\frac{|R'|}{|R_0|}\cdot p(\zeta',R')+ \frac{|R_1|}{|R|}\sum_{\substack{(\zeta',R')\in \mathcal{N}_{d+1} ^{\reff}\\ R'\subseteq R_1}}\frac{|R'|}{|R_1|}\cdot p(\zeta',R')\\
    &\leq \frac{|R_0|}{|R|}\left(p(\zeta,R)+\log_{2}\left(\frac{|R|}{|R_0|}\right)+2\right)+  \frac{|R_1|}{|R|}\left(p(\zeta,R)+\log_{2}\left(\frac{|R|}{|R_1|}\right)+2\right) \\
    &\leq p(\zeta,R)+ 3,
\end{align*}
where the second transition is by~\Cref{lem:reg}, and the last transition comes from the fact that the binary entropy is upper bounded by $1$.

\paragraph{Upper bounding $\adv(\Pi)$ by $\adv\left(\Pi^{\reff}\right)$.} By the definition of $\adv(\Pi)$, we have
    \begin{align*}
        2\cdot \adv(\Pi)\leq \sum_{R\in \mathcal{N}_{r}}|\Dyes(R) - \Dno(R)|.
    \end{align*}
    By construction, it is easy to see that each leaf rectangle $R$ of $\Pi$ is decomposed into several subrectangles by the refined protocol $\Pi^{\reff}$, and we have: 
    \begin{align*}
        2\cdot \adv(\Pi)\leq \sum_{R\in \mathcal{\mathcal{N}}_{r}}|\Dyes(R) - \Dno(R)|
        \leq \sum_{R\in \mathcal{\mathcal{N}}_{r}} \sum_{
        \substack{R'\subseteq R\\ (\zeta,R')\in \mathcal{N}_{r}^{\reff}}} |\Dyes(R') - \Dno(R')|
         = 2\cdot \adv(\Pi^{\reff}). 
    \end{align*}
    Here, the second inequality comes from the fact that each $R\in \mathcal{N}_{r}$ is the disjoint union of all $R'$ such that $R'\subseteq R$ and $(\zeta,R')\in \mathcal{N}_{r}^{\reff}$ for some $\zeta$. The last equality comes from the construction of $\Pi^{\reff}$ in lines 11 and 12. 
\end{proof}

\subsection{Analyzing Global Protocols: Preparatory Statements}
Armed with~\Cref{lem:arbitrary_to_global}, we would now like to
analyze global protocols. In this section, we give two statements that will be useful to us towards this goal.

The first statement captures the following intuition: 
for a $\zeta = \left(z^{(1)},\dots,z^{(K)}\right)$-global rectangle $R$, if $\bigcup_{i=1}^K \supp(z^{(i)})$ is small, then the probability that a random element from $R$ contains an edge between vertices appearing in $\zeta$ which doesn't already appear in one of the $z^{(i)}$'s, is small. 
Intuitively, this is true because there are at most $4|\zeta|^2$ such edges, and by definition of globalness none of them is particularly more prevalent in members of $R$.
\begin{lemma}\label{lem:low_cycle_probability}
    Suppose that $R$ 
    is a $\zeta = \left(z^{(1)},\dots,z^{(K)}\right)$-global rectangle where $|\zeta|= m\leq \alpha n/10$. Let $B\subseteq R $ be the subset consisting of all tuples $(y^{(1)},\dots,y^{(K)})\in R$,  such that $\bigcup_{i=1}^K \left(\mathrm{supp}(y^{(i)})\setminus \supp(z^{(i)})\right)$ contains at least one edge within $\bigcup_{i=1}^KN\left(\supp\left(z^{(i)}\right)\right)$.
    Then
    \begin{align*}        \mathcal{D}_{\mathrm{no}}(B) \leq \mathcal{D}_{\mathrm{no}}(R)\cdot \binom{2m}{2}\cdot \frac{8\alpha K}{n}.
    \end{align*}
\end{lemma}
\begin{proof}
   Denote $N = \bigcup_{i=1}^KN\left(\supp\left(z^{(i)}\right)\right)$, 
   and for each $i\in [K]$ define 
   \[
   B^{(i)} = \left\{ (y^{(1)},\ldots,y^{(K)})\in R~:~
   \supp\left(y^{(i)}\right)\setminus \supp\left(z^{(i)}\right)\text{ contains some edge from } N\right\}.
   \]
   Clearly $B\subseteq \bigcup_{i=1}^K B^{(i)} $, so it suffices to show that $\mathcal{D}_{\mathrm{no}}
    \left(B^{(i)}\right)\leq \mathcal{D}_{\mathrm{no}}(R)\cdot\binom{2m}{2}\frac{8\alpha }{n}$ for all $i\in [K]$. 

    Fix $i\in [K]$ and an edge $e=\{u,v\}$. 
    Using $z^{(i)}$-globalness of $A^{(i)}$, the fraction of $y^{(i)}\in A^{(i)}$ such that $e\in \supp\left(y^{(i)}\right)\setminus \supp\left(z^{(i)}\right)$ is upper bounded by
    \[
    2\cdot\frac{\left|\left\{y^{(i)}\in\Omega_{z^{(i)}}\mid e\in\supp(y^{(i)})\right\}\right|}{\left|\Omega_{z^{(i)}}\right|}
    \leq
    4\cdot \frac{\alpha n-|\supp(z^{(i)})|}{\binom{n-2|\supp(z^{(i)})|}{2}}
    \leq \frac{8\alpha}{n}.    
    \]
    Taking a union bound over all edges within $N$ completes the proof. 
\end{proof}
The second statement captures the idea that if a restriction contains no cycle and $R$ is global with respect to it, then 
the measure of $ R $ under $ \mathcal{D}_{\text{yes}}$ and $ \mathcal{D}_{\text{no}} $ remains nearly identical. 
This is the most technical lemma of this paper, and it is proved in Section \ref{sec:global_rectangle}.

\begin{restatable}[Discrepancy bound]{lemma}{discrepancylemma}\label{lem:main}
Let $ R = A^{(1)} \times \cdots \times A^{(K)} $ be a rectangle and suppose that for all $i$, $A^{(i)} $ is $z^{(i)}$-global. Suppose the following conditions hold for constants $\gamma, \eta\in (0,\frac{1}{10})$:
\begin{enumerate}
    \item The edge sets $\left( \supp(z^{(i)})\right)_{i\in [K]}$ are pairwise disjoint, and their union does not contain any cycle.
    \item $\sum_{i}|\supp(z^{(i)})| \leq \gamma n^{1/3}$.
    \item $|A^{(i)}|/|\Omega_{z^{(i)}}| \geq 2^{-\eta n^{1/3}}$ for all $ i $.
\end{enumerate}
Then:
\[
\frac{|\calD_{\no}(R)-\calD_{\yes}(R)|}{\calD_{\no}(R)}\leq \exp\left(K\sqrt{4\eta\gamma^{2}K^{2}+4\eta + 2\gamma}+o(1)\right)-1.
\]
\end{restatable}
\begin{proof}
    Deferred to~\Cref{sec:global_rectangle}.
\end{proof}

\subsection{Analyzing Global Protocols}
We are now ready to analyze global protocols. The following lemma asserts that any global protocol for $\mathsf{DIHP}(n,\alpha,K)$ with small number of rounds of communication has a small advantage.
\begin{lemma}\label{lem:analysis_of_global_protocols}
    Let $\alpha \in (0,10^{-7}]$ be a constant, and suppose that $r\leq 10^{-10}n^{1/3}/K^2$. If $\Pi$ is an $r$-round, global protocol for $\mathsf{DIHP}(n,\alpha,K)$, then $\adv(\Pi)< 0.1$.
\end{lemma}
\begin{proof}
Write $r = \gamma n^{1/3}$ so that $\gamma \leq 10^{-10}/K^2$. For each $0\leq d\leq r$ define $\mathcal{N}_{d}$  
as the set of pairs of restrictions and rectangles $(\zeta,R)$ that are obtained after $d$ rounds of communications according in $\Pi$. We further classify 
pairs $(\zeta,R)\in \mathcal{N}_{r}$ (which correspond to leaf nodes) into three types:
    \begin{enumerate}
        \item $\mathcal{R}_1$ is the collection of pairs $(\zeta, R)\in \mathcal{N}_{r}$ such that $p(\zeta, R)\geq 10^3 \gamma n^{1/3}$; 
        \item $\mathcal{R}_2$ is the collection of 
        pairs $(\zeta, R)\in \mathcal{N}_{r}$ such that $\bigcup_{i=1}^K \supp\left(z^{(i)}\right)$ contains a cycle or there exist two distinct indices $i,j$ such that $ \supp\left(z^{(i)}\right)\cap \supp\left(z^{(j)}\right)\neq \emptyset$; 
        \item $\mathcal{R}_3$ is the collection of 
        pairs $(\zeta, R)\in \mathcal{N}_{r}$ not in the first two types.
    \end{enumerate}
    We will prove the following properties regarding the partition $\mathcal{R}_1, \mathcal{R}_2, \mathcal{R}_3$:
    \begin{enumerate}
        \item $\sum_{(\zeta, R)\in \mathcal{R}_1} \Dno(R)\leq 0.01$;
        \item $\sum_{(\zeta,R)\in \mathcal{R}_2\setminus\mathcal{R}_1} \Dno (R)\leq 0.01$;
        \item $|\Dno(R)  - \Dyes (R)|\leq 0.01\cdot \Dno(R)$ for every $(\zeta, R)\in \mathcal{R}_3$.
    \end{enumerate}
    Before establishing these items, we quickly show how to conclude the proof of the lemma from them. Indeed, by~\Cref{lem:disc_easy} we have
    \begin{align*}
        2\cdot  \adv(\Pi^{\reff} )  &\leq \sum_{(\zeta,R)\in \mathcal{N}_{r}} |\Dno(R) - \Dyes(R)| \\
        &= \sum_{(\zeta,R)\in \mathcal{R}_1\cup \mathcal{R}_2} |\Dno(R) - \Dyes(R)| + \sum_{(\zeta, R)\in \mathcal{R}_3}  |\Dno(R) - \Dyes(R)|\\
        &\leq \sum_{(\zeta,R)\in \mathcal{R}_1\cup \mathcal{R}_2}\left(\Dyes(R)+\Dno(R)\right) + \sum_{(\zeta,R)\in \mathcal{R}_3}  |\Dno(R) - \Dyes(R)| \\
        & =  \sum_{(\zeta,R)\in \mathcal{R}_1\cup \mathcal{R}_2}\Dno(R) +1 - \sum_{(\zeta,R)\in \mathcal{R}_3}\Dyes(R) +\sum_{(\zeta,R)\in \mathcal{R}_3}  |\Dno(R) - \Dyes(R)|\\
        &\leq \sum_{(\zeta,R)\in \mathcal{R}_1\cup \mathcal{R}_2}\Dno(R) +1 - \sum_{(\zeta,R)\in \mathcal{R}_3}\Dno(R)  +2\cdot \sum_{(\zeta,R)\in \mathcal{R}_3}  |\Dno(R) - \Dyes(R)|\\
        &\leq 2\cdot \sum_{(\zeta,R)\in \mathcal{R}_1\cup \mathcal{R}_2}\Dno(R)+0.02\cdot \sum_{(\zeta,R)\in \mathcal{R}_3} \Dno(R) <0.1.
    \end{align*}
    In the sixth transition, we used $|\Dno(R)  - \Dyes (R)|\leq 0.01\cdot \Dno(R)$ for every $(\zeta, R)\in \mathcal{R}_3$, and in the last transition we used 
    $\sum_{(\zeta, R)\in \mathcal{R}_1} \Dno(R)\leq 0.01$
    and $\sum_{(\zeta, R)\in \mathcal{R}_2} \Dno(R)\leq 0.01$.
    We now move on the prove the properties of $\mathcal{R}_1,\mathcal{R}_2,\mathcal{R}_3$.
    
    \paragraph{Proving the property of $\mathcal{R}_3$:}
    For any $(\zeta,R)\in \mathcal{R}_3$, 
    writing $\zeta = \left(z^{(1)},\dots,z^{(K)}\right)$ and 
    $R= A^{(1)}\times \cdots\times A^{(K)}$, we get that as $(\zeta,R)\not\in \mathcal{R}_1$ we have 
    \[
    |\zeta|\leq p(\zeta,R)\leq 10^3\gamma n^{1/3}\leq 10^{-7}n^{1/3}/K^2,
    \qquad 
    \frac{|A^{(i)}|}{|\Omega_{z^{(i)}}|}\geq 2^{-p(\zeta, R)}\geq 2^{-10^{-7}n^{1/3}/K^2}.
    \]
    Also, as $(\zeta,R)\not\in \mathcal{R}_2$ 
    we get that $\bigcup_{i=1}^K \supp\left(z^{(i)}\right)$ does not contain any cycle, and for any two distinct indices $i,j\in [K]$ we have $\supp\left(z^{(i)}\right)\cap \supp\left(z^{(j)}\right)= \emptyset$.
    Using~\Cref{lem:main}, we get our desired bound of $$\frac{|\Dno(R)  - \Dyes (R)|}{ \Dno(R)}\leq \exp\left(K\sqrt{4\cdot (10^{-7}/K^2)^3K^2+6\cdot (10^{-7}/K^2)}\right)-1\leq 0.01.$$

    \paragraph{Upper bounding $\mathcal{R}_1$.} To upper bound the total weight of pairs $(\zeta, R)$ with $p(\zeta,R)\geq 10^3 \gamma n^{1/3}$, we first bound the weighted sum of potentials $p(\zeta,R)$ over all leaf pairs $(\zeta,R)\in \mathcal{N}_{r}$. More precisely, we show that
    \begin{align}\label{eq:weighted_sum}
        \sum_{(\zeta,R)\in \mathcal{N}_{r}} \frac{|R|}{|\Omega^{K}|}\cdot p(\zeta,R) \leq 3\cdot r,
    \end{align}
    and the proof proceeds by induction argument on the depth $d$. We prove that for all $d$,
    \begin{align}\label{eq:induction_hypo}
    \sum_{(\zeta,R)\in \mathcal{N}_d} \frac{|R|}{|\Omega^{ K}|}\cdot p(\zeta,R)\leq 3\cdot d.
    \end{align}
    When $d=0$ the statement is clear as $\mathcal{N}_0$ only
    contains the trivial rectangle $\Omega^K$, so it has potential equal to $0$.
    Let $d>0$ and assume that~\eqref{eq:induction_hypo} holds for $d-1$. We have
    \begin{align*}
        \quad\sum_{(\zeta',R')\in\mathcal{N}_d} \frac{|R'|}{|\Omega^{K}|}\cdot p(\zeta',R')  
        &= \sum_{(\zeta,R)\in \mathcal{N}_{d-1}} \frac{|R|}{|\Omega^{K}|}\sum_{R'\subseteq R,(\zeta',R')\in \mathcal{N}_d} \frac{|R'|}{|R|}\cdot p(\zeta',R') \\
        &\leq \sum_{(\zeta,R)\in \mathcal{N}_{d-1}} \frac{|R|}{|\Omega^{K}|} \cdot \left(p(\zeta,R)+3 \right)&\\
        &\leq 3(d-1) + 3 = 3d,
    \end{align*}
    where the first transition is by definition, the second one
    is by~\Cref{def:global_protocol}, and the last transition is by the inductive hypothesis. This completes the inductive step, and in particular establishes~\eqref{eq:weighted_sum}.
    
    The bound on $\mathcal{R}_1$ now follows by Markov's inequality applied on~\eqref{eq:weighted_sum}:
    \[
        \sum_{(\zeta,R)\in \mathcal{R}_1} \Dno(R) = \sum_{(\zeta,R)\in \mathcal{R}_1} \frac{|R|}{|\Omega^{K}|}\leq \frac{\sum_{(\zeta,R)\in \mathcal{N}_{r}}\frac{|R|}{|\Omega^{K}|}\cdot p(\zeta,R)}{10^3\gamma n^{1/3}}\leq \frac{3\cdot r}{10^3\cdot r}<0.005.
    \]
    \paragraph{Upper bounding $\mathcal{R}_2\setminus \mathcal{R}_1$.} 
    The key ingredient in the proof is the following observation: 
    \begin{align}\label{eq:cycle}
        \sum_{(\zeta,R)\in \mathcal{R}_2\setminus \mathcal{R}_1} \Dno(R) \leq \sum_{d=0}^{r-1}\sum_{\substack{\left(\zeta,R\right)\in\mathcal{N}_d\\ |\zeta|\leq 10^3\gamma n^{1/3}}} \Dno(R)\cdot \frac{8\alpha K}{n}\cdot \binom{2|\zeta|}{2}.
    \end{align}
    Indeed, once we have that, the right hand side is clearly
    at most $\frac{8\alpha K}{n}\cdot r\cdot \binom{2\cdot 10^3 \gamma n^{1/3}}{2} < 0.01$ using the conditions on $r$ and
    $\gamma$.
    
    We now move on to establish~\eqref{eq:cycle}. For every pair $\left(\zeta=(z^{(1)},\dots,z^{(K)}),R\right)$ such that $R$ is a $\zeta$-global rectangle, 
    let $N(\zeta) = \bigcup_{i=1}^KN\left(\supp\left(z^{(i)}\right)\right)$ and denote by $B(\zeta,R)\subseteq \Omega^{K}$ the set
    \begin{align*}
        \left\{(y^{(1)},\dots,y^{(K)})\in R: \bigcup_{i=1}^K\left(\supp\left(y^{(i)}\right)\setminus \supp\left(z^{(i)}\right)\right)\text{ contains an edge within }N(\zeta)\right\}.
    \end{align*}
    With this definition, we prove ~\eqref{eq:cycle} in two steps: 
    \begin{itemize}
        \item First, we prove that \begin{align}\label{eq:cycle_2}
        \sum_{(\zeta,R)\in \mathcal{R}_2\setminus \mathcal{R}_1} \Dno(R) \leq \sum_{d=0}^{r-1}\sum_{\substack{\left(\zeta,R\right)\in\mathcal{N}_d \\ |\zeta|\leq 10^3\gamma n^{1/3}}} \Dno(B(\zeta,R)).
        \end{align}
        \item Secondly, we use~\Cref{lem:low_cycle_probability} to get that $\Dno(B(\zeta,R))\leq \frac{8\alpha K}{n}\cdot \binom{2|\zeta|}{2} \cdot \Dno(R)$, and plug it into the right hand side of~\eqref{eq:cycle_2} to get~\eqref{eq:cycle}.
    \end{itemize}
    The rest of the argument is devoted to proving~\eqref{eq:cycle_2}, and we first establish the following claim. 
    \begin{claim}\label{claim:cycle}
    For every pair $(\zeta,R)\in \mathcal{R}_2\setminus \mathcal{R}_1$, there exist an integer $d'\in \{0,\dots,r-1\}$ and a pair $(\zeta',R')\in \mathcal{N}_{d'}$ with $|\zeta'|\leq 10^3\gamma n^{1/3}$ such that $R\subseteq B(\zeta',R')$. 
    \end{claim}
    \begin{proof}
        Fix a pair $(\zeta,R)\in \mathcal{R}_2\setminus \mathcal{R}_1$, so that by definition $(\zeta,R)\in \mathcal{N}_r$. It follows that for every integer $d''\in \{0,\dots,r\}$, there exists a pair $(\zeta'',R'')\in \mathcal{N}_{d''}$ such that $R\subseteq R''$ and $|\zeta''|\leq |\zeta|\leq 10^3\gamma  n^{1/3}$. We pick the smallest $d''\in \{0,\dots,r\}$ such that there exists a pair $(\zeta'',R'')\in \mathcal{N}_{d''}$ with the following properties: 
        \begin{enumerate}
            \item $R\subseteq R''$ and $|\zeta'|\leq |\zeta|\leq 10^3\gamma  n^{1/3}$;
            \item\label{property:have_cycle} writing $\zeta''=\left(z''^{(1)},\dots,z''^{(K)}\right)$, 
            we have that either (1) $\bigcup_{i=1}^K \supp\left(z''^{(i)}\right)$ contains a cycle, 
            or (2) there exists an edge $e$ such that $e\in \supp\left(z''^{(i)}\right)\cap \supp\left(z''^{(j)}\right)$ for two distinct indices $i,j\in[K]$. 
        \end{enumerate}
        With $d''$ and $(\zeta'',R'')$ in hand, we set $d' = d''-1$ and pick $(\zeta' = \left(z'^{(1)},\dots,z'^{(K)}\right),R')\in\mathcal{N}_{d'}$ to be the parent node of $(\zeta'',R'')$ in the communication tree. 
        We claim that the pair $(\zeta',R')$ satisfies the assertion of the claim. First we argue that $1\leq d''\leq r$, so that this parent node exists. 
        Indeed, $(\zeta, R)$ itself has the two properties 
        so $d''\leq r$, but the pair at the root node 
        $(\vec{0},\Omega^{K})\in \mathcal{N}_{0}$ does not so $d''>0$.   

        Next, we show that $R \subseteq B(\zeta', R')$, 
        and as $R\subseteq R''$ it suffices to show that $R''\subseteq B(\zeta',R')$.
        By the minimality of $d''$ we know that property \ref{property:have_cycle} does not hold for the pair $(\zeta', R')$. Furthermore, because exactly one player speaks in each round, there exists at most one coordinate $i\in[K]$ for which $z''^{(i)} \neq z'^{(i)}$. We consider two cases, according to which 
        condition in property~\ref{property:have_cycle} holds.
        \begin{enumerate}
            \item If $\bigcup_{j=1}^K \supp\left(z''^{(j)}\right)$ contains a cycle, 
            then that cycle must contain an edge from $\supp(z''^{(i)}) \setminus \supp(z'^{(i)})$ as otherwise $(\zeta',R')$
            would also satisfy property~\ref{property:have_cycle}. Calling that edge $e = \{u,v\}$, we note that as the support of $z''^{(i)}$ is a matching and $u,v$ are on a 
            cycle in $\bigcup_{j=1}^K \supp\left(z''^{(j)}\right)$, there must be $j,\ell\neq i$ such that $u$ appears in some edge in
            $\supp\left(z''^{(j)}\right) = \supp\left(z'^{(j)}\right)$
            and $v$ appears in some edge in
            $\supp\left(z''^{(\ell)}\right) = \supp\left(z'^{(\ell)}\right)$. 
            In particular, we get that $u,v\in N(\zeta')$, 
            so $e$ is an edge within $N(\zeta')$.
            Noting that any $(y''^{(1)},\dots, y''^{(K)})\in R''$ has $e\in \supp(y''^{(i)})\setminus \supp(z'^{(i)})$, we get that 
            $(y''^{(1)},\dots, y''^{(K)})\in B(\zeta',R')$.
            \item Else, there are $j\neq \ell$ and an edge 
            $e = \{u,v\}\in \supp(z''^{(j)})\cap \supp(z''^{(\ell)})$. As property~\ref{property:have_cycle} fails for 
            $(\zeta',R')$ it must be the case that either $j=i$
            or $\ell = i$, say without loss of generality that
            $\ell = i$ so that $e\not\in \supp(z'^{(i)})$. As $e\in \supp(z''^{(j)}) = \supp(z'^{(j)})$ we get that $u,v\in N(\zeta')$
            so $e$ is an edge within $N(\zeta')$.
            Noting that any $(y''^{(1)},\dots, y''^{(K)})\in R''$ has $e\in \supp(y''^{(i)})\setminus \supp(z'^{(i)})$, we get that 
            $(y''^{(1)},\dots, y''^{(K)})\in B(\zeta',R')$.\qedhere
        \end{enumerate}
    \end{proof}
    We finish by noting that~\Cref{claim:cycle} implies~\eqref{eq:cycle_2}. Indeed, noting that the 
    for distinct rectangle-restriction pairs $(\zeta_1,R_1), (\zeta_2,R_2)\in \mathcal{R}_2\setminus \mathcal{R}_1$ we have $R_1\cap R_2=\emptyset$, we get that their 
    contributions on the right hand side of~\eqref{eq:cycle_2} is disjoint so no overcounting occurs.        
\end{proof}

\subsection{Lower Bound for \textsf{DIHP}}
Now, we are ready to prove~\Cref{thm:main}, restated below.
\thmmain*
\begin{proof}
    Assume we have a communication protocol $\Pi$ for $\textsf{DIHP}(n,\alpha, K)$ such that $|\Pi| =r \leq 10^{-10}n^{1/3}/K^2$. By ~\Cref{lem:arbitrary_to_global} we get that there exists a $r$-round global protocol $\Pi^{\reff}$ such that $\adv(\Pi)\leq \adv(\Pi^{\reff})$, and by~\Cref{lem:analysis_of_global_protocols} we have $\adv(\Pi^{\reff})\leq 0.1$. Thus, 
    $\adv(\Pi)\leq 0.1$, and the proof is concluded.
\end{proof}

\begin{remark} 
A more refined analysis yields a $\Omega(\sqrt{n})$ lower bound for the search version of the $\textsf{DIHP}$ problem, as studied in~\cite{CKP+23}. Roughly speaking, the $\Theta(n^{1/3})$ bottleneck 
in the above argument appears both when we apply~\Cref{lem:main} as well as when we use~\eqref{eq:cycle_2} to bound the contribution
of $\mathcal{R}_2\setminus\mathcal{R}_1$. 
As we do not know how to improve~\Cref{lem:main} beyond $\Theta(n^{1/3})$, this is the limit of our argument. 
However,~\Cref{lem:main} is not necessary when studying the search version of the problem, and in this context the only bottleneck 
is the use of~\eqref{eq:cycle_2}. A more 
careful analysis may yield a lower bound of $\Omega(\sqrt{n})$
in this context.

\end{remark}

\section{Bounding the Discrepancy of Global Rectangles}\label{sec:global_rectangle}
In this section we prove~\Cref{lem:main} assuming a global hypercontractivity result, which we prove later in~\Cref{sec:global_hypercontractivity}. 

\subsection{The Induced Distribution over Bipartitions}
Recall that a node in the protocol tree corresponds to 
a rectangle $R = A^{(1)}\times\ldots\times A^{(K)}$, 
where $A^{(i)}\subseteq \Omega$ is the subset of
possible inputs for player $i$. A natural question
is what information do players $j\neq i$ learn about 
the bipartitions $x\in \mathbb{F}_2^U$ that are 
consistent with player $i$ based on $A^{(i)}$. These type of questions come 
up in the proof of~\Cref{lem:main}, and in this section 
we formally define the corresponding probability distribution
over bipartition and state a result about it.

To begin the discussion, note that each element $y\in \Omega^{U,m}$ corresponds to a labeled set of edges on $U$, 
giving rise to bipartitions $x\in\mathbb{F}_2^U$ that are consistent with it, which we denote by $L(y)$:
\begin{definition}\label{def:L(y)}
    For every string $y\in \{-1,0,1\}^{\binom{U}{2}}$, we define the affine subspace $L(y)\subseteq \mathbb{F}_{2}^{U}$ by
    \[
    L(y):=\bigcap_{\{u,v\}:\,y_{uv}=1}\left\{x\in\mathbb{F}_{2}^{U}:x_{u}+ x_{v}=0\right\}\cap\bigcap_{\{u,v\}:\,y_{uv}=-1}\left\{x\in\mathbb{F}_{2}^{U}:x_{u}+ x_{v}=1\right\}.
    \]
    When used in this context, we call the string $y$ a ``constraint''.
\end{definition}
We will often consider constraints $y$ whose supports are matchings, i.e., such that $y\in\Omega^{[n],\alpha n}$. In that case we note that $L(y)$ is a subspace of co-dimension $\alpha n$ and so $|L(y)|=2^{n-\alpha n}$. It is easy to see that each bipartition $x\in L(y)$ is consistent with $y$, and furthermore
if we sample $(y^{(1)},\ldots,y^{(K)})\sim \mathcal{D}_{\text{yes}}$ conditioned on $y^{(1)} = y$, then
the distribution of $x$ is uniform over $L(y)$.

We may now associate a probability distribution over $\mathbb{F}_2^U$ with a given subset $A\subseteq \Omega = \Omega^{U,m}$:
\begin{definition}\label{def:conditional-probability}
    For a subset $A\subseteq \Omega^{U,m}$ and $x\in\mathbb{F}_{2}^{U}$, we define
    \[
    \Pr[x|A]:=\frac{1}{|A|}\sum_{y\in A}\frac{\mathbbm{1}\{x\in L(y)\}}{|L(y)|}.
    \]
\end{definition}
In words, each $y\in A$ induces the uniform distribution over $L(y)$, and the distribution associated with $A$ is the uniform mixture of all of these.

We will be interested in studying the distribution of $x|A$ 
in the case that $A$ is global, and to gain some intuition 
we first consider the case that the set $A\subseteq \Omega^{U,m}$ is chosen
randomly by including each element with probability $p$. 
In that case, one can show (using standard concentration 
bounds) that a typical $x\in \mathbb{F}_2^{U}$ lies in 
$(1\pm o(1))2^{-m}\cdot p|\Omega^{U,m}|$ many of the $L(y)$'s, in
which case the distribution of $x|A$ is close to being uniform.
We will show that something close in spirit holds in the case that $A$ is global (instead of random). 
Our notion of closeness is with respect to $K$-norms, where $K$ is fairly large (this is the number of players). This type of closeness is stronger than the more standard statistical distance closeness, and it is required in our application. The statement below is essentially such a statement, except that it applies in a slightly more general case that includes restrictions:
\begin{lemma}\label{lem:norm}
Let $A\subseteq \Omega_{z'}$ be a $z'$-global set, 
let $z$ be a constraint that subsumes $z'$ and define the function $h:L(z)\rightarrow \mathbb{R}$ by  $h(x)= 2^{n - |\supp(z')|} \Pr[x|A]-1$. 
Suppose that the following conditions hold for constants $\gamma,\eta\in (0,\frac{1}{10})$:
\begin{enumerate}
    \item $\supp(z)$ does not contain any cycles; 
    \item $|\supp(z)|\leq \gamma n^{1/3}$; 
    \item $|A|/|\Omega_{z'}|\geq 2^{-\eta n^{1/3}}$.
\end{enumerate}
Then we have that $\lVert h \rVert_K \leq \sqrt{4\eta\gamma^{2}K^{2}+4\eta + 2\gamma}+o(1)$. Here, the $K$-norm of $h$ is with respect to the uniform distribution on the subspace $L(z)$, i.e., 
$\lVert h \rVert_K = \left(\E_{x\in L(z)}|h(x)|^K\right)^{1/K}$.
\end{lemma}
\begin{proof}
    Deferred to~\Cref{sec:proof_of_norm_bounds}.
\end{proof}

\subsubsection{Bounding Discrepancy via $K$-norms}\label{sec:reduce_to_norm}
Before beginning the proof of~\Cref{lem:norm}, we first
show how it can be used to prove~\Cref{lem:main}. The following lemma shows that for every rectangle $R$ in $\Omega^{K}$, one may relate $\mathcal{D}_{\mathrm{yes}}(R)$ and $\mathcal{D}_{\mathrm{no}}(R)$ via the probability distribution associated with the components of $R$.

\begin{lemma}\label{lem:relation_between_yes_no}
    For $\Omega = \Omega^{[n],\alpha n}$ and for any rectangle $R = A^{(1)}\times \cdots\times A^{(K)}\subseteq \Omega^{K}$, we have
    \begin{align*}
        \mathcal{D}_{\mathrm{yes}}(R) &= \mathcal{D}_{\mathrm{no}}(R) \cdot 2^{(K-1)n}\cdot \sum_{x\in\F_{2} ^n}\prod_{i=1}^K \Pr[x|A^{(i)}].
    \end{align*}
\end{lemma}
\begin{proof}
    Recall that in the sampling process of $\mathcal{D}_{\mathrm{yes}}$, we first uniformly sample a vector $x\in \mathbb{F}_{2}^n$, and for each $i\in [K]$, sample matchings $M^{(i)}$ of size $\alpha n$ independently and uniformly. The labeling $y^{(i)}$ is then 
    the unique $y^{(i)}\in\Omega$ with $\supp(y^{(i)})=M^{(i)}$ and $x\in L(y^{(i)})$. 
    Therefore, conditioned on $x\in\mathbb{F}_{2}^{n}$, the distribution of each $y^{(i)}$ is uniform over all $y\in \Omega$ such that $x\in L(y)$, and the number of such $y$ is equal to the number of matchings of size $\alpha n$, which is $|\Omega|/2^{\alpha n}$. Therefore, we have 
    \begin{align*}
        \mathcal{D}_{\mathrm{yes}}(R) = \sum_{x\in \mathbb{F}_{2}^n} \frac{1}{2^n} \prod_{i=1}^K \frac{\sum_{y\in A^{(i)}}\mathbbm{1}\{x\in L(y)\}}{|\Omega|/2^{\alpha n}}
        &=\sum_{x\in \mathbb{F}_{2}^n}\frac{1}{2^{n}}\prod_{i=1}^{K}\left(\frac{2^{n}\cdot |A^{(i)}|}{|\Omega|}\cdot \Pr[x|A^{(i)}]\right) \\
        &=\left(\prod_{i=1}^K \frac{|A^{(i)}|}{|\Omega|} \right) \cdot 2^{(K-1)n}\cdot\sum_{x\in \mathbb{F}_{2}^n}\prod_{i=1}^K\Pr[x|A^{(i)}]\\
        &=\mathcal{D}_{\mathrm{no}}(R) \cdot 2^{(K-1)n}\cdot \sum_{x\in\F_{2}^n}\prod_{i=1}^K \Pr[x|A^{(i)}],
    \end{align*}
    where in the second transition we used~\Cref{def:conditional-probability} and 
    $|L(y)|=2^{n-\alpha n}$.
\end{proof}

We now prove~\Cref{lem:main}, restated below.
\discrepancylemma*
\begin{proof}
    By Lemma \ref{lem:relation_between_yes_no}, it suffices to show 
    \begin{align*}
        \left|1- 2^{n(K-1)}\cdot \sum_{x\in \mathbb{F}_2^2}\prod_{i=1}^K \Pr[x|A^{(i)}]\right|\leq\exp\left(K\sqrt{4\eta\gamma^{2}K^{2}+4\eta + 2\gamma}+o(1)\right)-1 . 
    \end{align*}
    Define $z\in\{-1,0,1\}^{\binom{[n]}{2}}$ by 
    \[
    z_{uv}=\begin{cases}
    z^{(i)}_{uv}, &\text{if }\{u,v\}\in\supp(z^{(i)})\text{ for some }i\in[K],\\
    0, &\text{if }\{u,v\}\not\in\supp(z^{(i)})\text{ for all }i\in[K].
    \end{cases}
    \]
    We note that $z$ is well defined since the supports in the sequence $\left(\supp(z^{(i)})\right)_{i\in[K]}$ are disjoint. 
    Also, using this fact and the fact their union does not contain any cycles, we have 
    \begin{equation}\label{eq:yes-no-dim-relation}
    \dim(L(z))=n-\left|\supp(z)\right|=n-\sum_{i=1}^{K}\left|\supp(z^{(i)})\right|=-(K-1)n+ \sum_{i=1}^{K}\dim\left(L(z^{(i)})\right).
    \end{equation}
    Since $L(z)=\bigcap_{i=1}^{K}L(z^{(i)})$ and each distribution $\Pr[x|A^{(i)}]$ is supported on $L(z^{(i)})$, we have
    \begin{equation}\label{eq:yes-no-support-property}
    \sum_{x\in \mathbb{F}_{2}^n}\prod_{i=1}^K \Pr[x|A^{(i)}] =   \sum_{x\in L(z)}\prod_{i=1}^K \Pr[x|A^{(i)}].
    \end{equation}
    Defining $h_{i}:L(z)\rightarrow\mathbb{R}$ as in~\Cref{lem:norm} by $h_{i}(x):=\left|L(z^{(i)})\right|\cdot \Pr[x|A^{(i)}]-1$, we get
    
    \begin{align*}
    \left|1-2^{n(K-1)}\cdot \sum_{x\in\mathbb{F}_{2}^n}\prod_{i=1}^K \Pr[x|A^{(i)}]\right|
    &=\left|1-\frac{1}{\left|L(z)\right|}\sum_{x\in L(z)}\prod_{i=1}^{K}(1+h_{i}(x))\right|&(\text{using \eqref{eq:yes-no-dim-relation} and \eqref{eq:yes-no-support-property}})\\
    &=\left|\sum_{T\subseteq [K],\, T\neq \emptyset }\E_{x\in L(z)}\left[\prod_{i\in T}h_{i}(x)\right]\right|\\
    &\leq \sum_{T\subseteq [K],\, T\neq \emptyset }\E_{x\in L(z)}\left|\prod_{i\in T}h_{i}(x)\right|,
    \end{align*}
    where the last transition is by the triangle inequality. 
    Using H\"{o}lder's inequality, the last expression is at most
    \begin{align*}
    \sum_{\substack{T\subseteq [K]\\ T\neq\emptyset }}
    \prod_{i\in T}\|h_{i}\|_{|T|}
    \leq \sum_{\substack{T\subseteq [K]\\ T\neq\emptyset }}\prod_{i\in T}\|h_{i}\|_{K}
    =\prod_{i=1}^{K}\left(1+\|h_{i}\|_{K} \right)-1
    \leq \exp\left(K\sqrt{4\eta\gamma^{2}K^{2}+4\eta + 2\gamma}+o(1)\right)-1,
    \end{align*}
    where the last transition is by the inequality $1+s\leq \exp(s)$ and by~\Cref{lem:norm}.
\end{proof}

\subsection{Norm Bounds via Decay of Fourier Coefficients}\label{sec:proof_of_norm_bounds}
In this section we present tools that will be useful in the proof of~\Cref{lem:norm}. 

We begin by giving some high level motivation. The domain of the function $h$ in~\Cref{lem:norm} is a linear space $L(z)$ over $\mathbb{F}_2$, hence it can be identified with a Boolean cube $\mathbb{F}_2^{n'}$. With this in mind, the function $h$ 
can be thought of as being defined over $\mathbb{F}_2^{n'}$ and therefore it may be expanded to its discrete Fourier transform (see~\cite{o2014analysis}). In particular we may consider its degree decomposition 
$h = \sum_{i=0}^{n'}h^{=i}$. 
The analysis of the contribution of small $i$'s and large $i$'s 
is done in a similar way, but we carry it out separately. Indeed, while 
both bounds boil down to Fourier coefficients estimates,
the argument for large $i$'s requires one additional trick
(in the form of an application of H\"{o}lder's inequality). 
Hence we focus the discussion on small $i$'s.

For small $i$'s, we may use the (standard) hypercontractive inequality to argue that $\|h^{=i}\|_{K}$ is comparable to
$\|h^{=i}\|_{2}$. The benefit of $2$-norms is that, using Parseval's equality, they can be expressed using the Fourier coefficients of $h$ of level $i$. Thus, upper bounding the $K$-norm 
of $h^{=i}$ amounts to sufficiently good understanding of the 
degree $i$ Fourier coefficients of $h$. 

Establishing these 
Fourier coefficients bounds is the main content of Lemma~\ref{lem:global-decay}. While the proof of that lemma is the main subject of~\Cref{sec:global_hypercontractivity}, it is useful to have a high level picture of how that lemma is proved. Ideally, we would
have liked to conclude such bounds by appealing to the level $i$ inequality, which asserts that a Boolean-valued function with small average has small level $i$ Fourier weight (see~\cite[page 259]{o2014analysis} for example). However, the function $h$ in our case is not Boolean-valued, so there is no such level $i$ inequality for $h$. Instead of appealing to the level $i$ inequality directly, we relate the Fourier coefficients of 
$h$ to (sort of) Fourier coefficients of $1_A\colon \Omega\to\{0,1\}$, and then prove a variant of the level $i$ inequality for functions over $\Omega$. 

We remark that to do the relation between the Fourier coefficients of $h$ and $1_A$ effectively, namely for degree $i$ Fourier coefficients of $h$ to be related to degree $i$ Fourier coefficients of $1_A$, it is important 
that the identification between $L(z)$ and a Boolean cube $\mathbb{F}_2^{n'}$ is meaningful (as opposed to arbitrary). This motivates our discussion of partitions $B$ 
and subspaces $V^{B,b}$ defined by them. 
Next, a closer inspection of the function $h$ in~\Cref{lem:norm} suggests to view it as a function over $L(z')\supseteq L(z)$, as it is more related to the domain where $A$ lives. Indeed, our arguments
naturally give Fourier coefficients bounds for that function, and we then need to translate them to
the restriction of the function to $L(z)$. Towards this end, 
using the fact that $L(z) \subseteq L(z')$ and using their identification with Boolean cubes, we conclude a meaningful embeddings between the Boolean cubes they are identified with.
This provides us relations between the Fourier coefficients of $h$ as a function over $L(z')$ and its Fourier coefficients as a function over $L(z)$, and allows us to translate Fourier decay bounds from the former function to the latter function. This motivates our discussion of
``refinements'' and ``unrefinements'' of partitions, and 
of their effect on Fourier decay bounds.

\subsubsection{The Fourier Analytic Setup}
To present a meaningful identification of 
the subspace $L(z)$ with a Boolean cube, we 
give a different way of presenting 
subspaces of the form of $L(z)$.

\begin{definition}\label{def:subspace-using-partitions}
Let $B=(B_{1},\dots,B_{k})$ be a partition of $[n]$, and let $b=(b_{1},\dots,b_{n})\in\F_{2}^{n}$. We define the affine subspace $V^{B,b}\subseteq\mathbb{F}_{2}^{n}$ by
\[
V^{B,b}:=\bigcap_{\ell=1}^{k}\left\{x\in\mathbb{F}_{2}^{n}:x_{i}+b_{i}=x_{j}+b_{j},\;\forall i,j\in B_{\ell}\right\}.
\]
\end{definition}
We note that not every affine subspaces of $\F_{2}^{n}$ can be represented by a partition and a string as in~\Cref{def:subspace-using-partitions}. However, 
it is easy to see that for any constraint $z\in\{-1,0,1\}^{\binom{[n]}{2}}$, the subspace $L(z)$ can indeed be expressed in this form.\footnote{In fact, it is easily seen that the family of subspaces that can be expressed as in~\Cref{def:L(y)} coincides with the family of subspaces that can be expressed as in~\Cref{def:subspace-using-partitions}.} 
The main benefit of working with~\Cref{def:Fourier-on-partition-subspaces} 
is that it naturally gives rise to a canonical identification
between $V^{B,b}$ and a Boolean cube, and thus with a canonical
Fourier basis for functions over $V^{B,b}$.
\begin{definition}\label{def:canonical_map}
   Define the canonical identification map 
$\id:V^{B,b}\rightarrow \mathbb{F}_{2}^{k}$ that maps $x\in V^{B,b}$ to  $z\in\mathbb{F}_{2}^{k}$ defined by $x_{i}+b_{i}=z_{\ell}$, where $\ell\in [k]$ is such that $i\in B_{\ell}$.  
\end{definition}
We note that $\id$ is well defined, as by definition the value of 
$x_i+b_i$ is the same for all $i\in B_{\ell}$. We also note that $\id$ is a $1$-to-$1$ map, as $x$ can be recovered from $z$ (when $b$ and $B$ are thought of as fixed). Finally, we note that as the 
sets $B_1,\ldots,B_k$ are disjoint, sampling $x\in V^{B,b}$ uniformly, 
the distribution of $\id(x)$ is uniform over $\mathbb{F}_2^k$. Hence it
makes sense to define the Fourier basis of the space $V^{B,b}$ using
the Fourier basis over $\mathbb{F}_2^k$.
\begin{definition}\label{def:Fourier-on-partition-subspaces}
For a subset $S\subseteq [k]$, we define the character function $\chi_{S}:V^{B,b}\rightarrow\{-1,1\}$ by
\[
\chi_{S}(x):=\prod_{\ell\in S}(-1)^{(\id(x))_{\ell}}.
\]
\end{definition}
\begin{definition}
    For a function $f:V^{B,b}\rightarrow \mathbb{R}$ and a subset $S\subseteq [k]$, we define the corresponding Fourier coefficient of $f$ by
    \[
    \widehat{f}(S):=\frac{1}{2^{k}}\sum_{x\in V^{B,b}}f(x)\chi_{S}(x).
    \]
\end{definition}

Thus, the Fourier expansion of a function $f:V^{B,b}\rightarrow \mathbb{R}$ is given as $f(x) = \sum\limits_{S\subseteq [k]}\widehat{f}(S)\chi_S(x)$. Fourier analysis on $L(z)$ is in essence identical to the analysis on Boolean cubes $\F_{2}^{k}$. We will therefore adopt many of the notations therein, and
specifically define the degree $d$ part of a function $f$ as $f^{=d}(x) = \sum\limits_{|S|=d}\widehat{f}(S)\chi_S(x)$.

\subsubsection{Decay of Fourier Coefficients and the Main Decay Lemma}
The following definition of decay of Fourier coefficients is specifically designed for our context. 

\begin{definition}
Let $w$ be a real number in the range $(0,n)$, and let $c$ be a positive real number. 
We say a function $f:\mathbb{F}_{2}^{n}\rightarrow\mathbb{R}$ is $(w,\delta,c)$-decaying if 
\begin{enumerate}
\item $\widehat{f}(\emptyset)^{2}\leq \delta$,
\item for every $d\geq 1$, $f^{=2d-1}=0$, and
\item for every $1\leq d\leq n/2$, $\left\|f^{=2d}\right\|_{2}^{2}\leq c^{-d}F(n,d,w)$, where $F(n,d,w)$ is defined by
$$F(n,d,w)= 
\begin{cases}
\left(\frac{w}{n}\right)^{d}, &\text{if }0\leq d\leq w,\\
\left(\frac{d}{4n}\right)^{d}\cdot 2^{2w}, &\text{if }d>w.
\end{cases}
$$
We remark that for $d=0$ the value $F(n,d,w)$ is defined to be 1. Although this case is irrelevant for the definition of decaying functions, adopting this convention will be useful for subsequent analysis.
\end{enumerate}
\end{definition}

With setup, we now state the main decay lemma, asserting that 
if $A$ is global, then (an appropriate normalization of) $\Pr[x|A]$ is decaying:
\begin{restatable}{lemma}{globalimpliesdecay}\label{lem:global-decay}
Suppose that $U$ has $|U|\geq 10^{7}m$ where $m\geq 10(w+1)$, and let $A\subseteq \Omega^{U,m}$ be a global set with $|A|= 2^{-w}\cdot \left|\Omega^{U,m}\right|$. Then the function $f:\mathbb{F}_{2}^{U}\rightarrow \mathbb{R}$ defined by $f(x):=2^{|U|}\cdot\Pr[x|A]-1$ is $(w/2,0,2)$-decaying.
\end{restatable}
\begin{proof}
    Deferred to~\Cref{sec:global_hypercontractivity}.
\end{proof}

\subsubsection{Useful Properties of $F(n,d,w)$}
We will need the following proposition ensuring that the piecewise definition of the Fourier weight bound $F(n,d,w)$ is well-behaved and convenient to work with.
\begin{proposition}\label{prop:property-of-F}
The function $F(n,d,w)$ has the following properties:
\begin{enumerate}[label=(\arabic*)]
\item For fixed $n$ and $d$, the bound $F(n,d,w)$ is increasing in $w$, and for all $t\geq 1$
$$\frac{F(n,d,tw)}{F(n,d,w)}\geq t^{\min\{d,w\}}.$$
\item For fixed $n$ and $w<n$, the bound $F(n,d,w)$ is decreasing in $d$, and
$$\frac{F(n,d,w)}{F(n,d-1,w)}\leq \frac{\max\{d,w\}}{n}.$$ 
\end{enumerate}
\end{proposition}

\begin{proof}
We begin with the first item, and we fix $n$ and $d$. 
The function $F(n,d,w)$ is continuous and piecewise differentiable in $w$, in the range $w\in(0,n)$. We have
\begin{align*}
\frac{\d}{\d w}\ln F(n,d,w)&=\begin{cases}
\ln 4, &\text{if }0<w< d\\
d/w, &\text{if }w>d 
\end{cases}\\
&\geq \min\{d,w\}/w.
\end{align*}

It follows that
\[
\int_{w}^{tw}\frac{\d}{\d r}\ln F(n,d,r) dr
\geq 
\int_{w}^{tw}\frac{\min\{d,r\}}{r} dr
\geq
\int_{w}^{tw}\frac{\min\{d,w\}}{r} dr
=
\min\{d,w\}\cdot \ln t,
\]
which translates to $F(n,d,tw)/F(n,d,w)\geq t^{\min\{d,w\}}$, giving
the first item. 

For the second item, we have 
\[
\frac{F(n,d,w)}{F(n,d-1,w)}\leq \max\left\{\frac{w}{n},\frac{d^d}{4(d-1)^{d-1}\cdot n}\right\}\leq \frac{\max\{d,w\}}{n},
\]
as desired.
\end{proof}

\subsubsection{Fourier Decay implies $K$-norm Bounds}
The following results shows that fast decay of Fourier coefficients coupled with mild bounds on the $L^{\infty}$-norm implies good bounds for the $K$-norm of a function.
\begin{lemma}\label{lem:decay-to-K-norm}
Let $K\geq 2$ be an integer and let $\gamma\in(0,\frac{1}{4})$ be a constant. Suppose $h:\mathbb{F}_{2}^{k}\rightarrow\mathbb{R}$ is $(\gamma k,\delta,1)$-decaying, and that $\left\|h\right\|_{\infty}\leq 2^{o(k)}$. Then $\|h\|_{K}\leq \sqrt{\delta+2\gamma K^{2}}+o_{k}(1)$, where $o_k(1)$ hides a term that tends to 0 as $k\rightarrow+\infty$.
\end{lemma}
\begin{proof}
By classical hypercontractivity (see~\cite{o2014analysis}, proof of Theorem 9.21), we have
\[
\left\|h^{\leq 2\gamma k}\right\|_{K}\leq \left(\sum_{d=0}^{\lfloor \gamma k\rfloor }(K-1)^{2d}\left\|h^{=2d}\right\|_{2}^{2}\right)^{1/2}
\leq \left(\delta+\sum_{d=1}^{\lfloor \gamma k\rfloor }K^{2d}\gamma^{d}\right)^{1/2}
\leq \sqrt{\delta+2\gamma K^{2}}.
\]
Using Parseval and~\Cref{prop:property-of-F} we also get
\[
\left\|h^{> 2\gamma k}\right\|_{2}\leq\left(\sum_{d=\lceil \gamma k\rceil}^{k/2}F(k,d,\gamma k)\right)^{1/2}\leq \sqrt{2F(k,\lceil \gamma k\rceil,\gamma k)}\leq \sqrt{2\gamma^{\gamma k}}.
\]
Therefore, using H\"{o}lder's inequality and Cauchy-Schwarz we get that
\begin{align*}
\left\|h\right\|_{K}^{K}=\left\langle h^{K-1}, h\right\rangle
&= \left\langle h^{K-1}, h^{\leq 2\gamma k}\right\rangle +
\left\langle h^{K-1}, h^{>2\gamma k}\right\rangle\\
&\leq \left\|h^{K-1}\right\|_{K/(K-1)}\left\|h^{\leq 2\gamma k}\right\|_{K}+\left\|h^{K-1}\right\|_{2}\left\|h^{> 2\gamma k}\right\|_{2}\\
&\leq \sqrt{\delta+2\gamma K^{2}}\left\|h\right\|_{K}^{K-1}+\sqrt{2\gamma^{\gamma k}}\|h\|_{\infty}^{K-1}\\
&\leq \sqrt{\delta+2\gamma K^{2}}\left\|h\right\|_{K}^{K-1}+o_{k}(1),
\end{align*}
and the conclusion follows by rearranging.
\end{proof}

\subsubsection{Refinements, Unrefinements and Fourier Decay}
In this section we present the notion of unrefinement and study its effect on Fourier decay bounds.

Suppose we have $z,z'\in \{-1,0,1\}^{\binom{[n]}{2}}$ where $z$ subsumes $z'$, 
and suppose we have some function $g\colon L(z')\to \mathbb{R}$ 
for which we managed to prove that it has a strong Fourier decay.
What can we say about the function $h|_{L(z)}$? Does it 
also have a Fourier decay? Such questions arise in our 
argument for~\Cref{lem:norm}, and in this section we 
develop some machinery to answer to them.

Recall that the definition of 
the Fourier transform proceeds by thinking of $L(z')$ as a space $V^{B',b}$, and then considering the cube given by the image of the map $\id$. Each block $B_i'$ then corresponds to a coordinate, and the Fourier expansion is defined correspondingly. Thinking of $\supp(z')$ as a graph, the blocks $\{B_i'\}$ are its connected components. 
To get $L(z)$ into the picture we need to explain how its representation relates to $V^{B',b}$. The simplest example for $z$ that subsumes $z'$ is $z$ that is identical to $z'$, except that $z_{uv}$ is non-zero for some $\{u,v\}\not\in \supp(z')$. 
In that case, the graph of $z$ is the graph of $z'$ with 
one additional edge, which may lead to a merge of two of the connected component of $z'$. More generally, the graph of $z$ that 
subsumes $z'$ is the graph of $z$ with additional edges, which may lead to merges between connected components of $z'$s. 
In other words, the connected components of $z$ are unions of $B_i'$'s, and this is what refer to 
as ``unrefinement''. The unrefinement operation 
gives a representation of $L(z)$ as the space $V^{B,b}$ where
each $B_i$ is a union of possibly several $B_j'$'s. In the other direction, $B'$ may be viewed as a refinement of $B$.

The following result asserts that for a function $g$ as above that has Fourier decay, applying a mild unrefinement
operation one gets a function that also has a decent Fourier decay.
\begin{lemma}\label{lem:override-constraint}
Suppose $B$ and $B'$ are partitions of the set $[n]$, and $b\in\F_{2}^{n}$ is a string of bits. Suppose $B'$ is a refinement of $B$, that is, every component of $B$ is a union of components of $B'$. Let $\gamma,\eta\in (0,\frac{1}{10})$, and suppose that 
\[
\sum_{\ell}|B_{\ell}|\cdot\mathbbm{1}\{|B_{\ell}|\geq 2\}\leq \gamma n^{1/3}.
\]
If $g:V^{B',b}\rightarrow\mathbb{R}$ is $\left(\eta n^{1/3}, 0 ,1\right)$-decaying, then $h=g|_{V^{B,b}}$ is $\left(4\eta\gamma^{2}|B|,8\eta + 2\gamma,1\right)$-decaying. 
\end{lemma}
\begin{proof}
Suppose $|B|=k$ and $|B'|=k'$. Note that the conditions imply $k'\geq k\geq n/2$. For any $S\subseteq[k]$, we define $\mathcal{T}(S)$ to be the collection of subsets $T\subseteq[k']$ such that: 
\begin{enumerate}
\item for every $\ell\in S$, the number of elements $j\in T$ such that $B'_{j}\subseteq B_{\ell}$ is odd;

\item for every $\ell\in[k]\setminus S$, the number of elements $j\in T$ such that $B'_{j}\subseteq B_{\ell}$ is even.
\end{enumerate}
We will first show that the Fourier coefficient $\widehat{h}(S)$ is equal
to the sum of the coefficients $\widehat{g}(T)$ for $T\in\mathcal{T}(S)$. Towards this end, 
note that if $\chi'_{T}:V^{B',b}\rightarrow \mathbb{R}$ is character 
on $V^{B',b}$ for $T\subseteq[k']$, and $\chi_{S}:V^{B,b}\rightarrow \mathbb{R}$ is the character on $V^{B,b}$ for $S\subseteq[k]$, then
$$\sum_{x\in V^{B,b}}\chi'_{T}(x)\chi_{S}(x)=
\begin{cases}
2^{k},&\text{if }T\in \mathcal{T}(S),\\
0, &\text{if }T\not\in\mathcal{T}(S).
\end{cases}$$
Therefore,
\begin{align}\label{eq:refine_fourier}
\widehat{h}(S)
=\frac{1}{2^{k}}\sum_{x\in V^{B,b}}h(x)\chi_{S}(x)
=\frac{1}{2^{k}}\sum_{x\in V^{B,b}}g(x)\chi_{S}(x)
&=\frac{1}{2^{k}}\sum_{x\in V^{B,b}}\sum_{T\subseteq[k']}\widehat{g}(T)\chi'_{T}(x)\chi_{S}(x)\notag\\
&=\frac{1}{2^{k}}\sum_{T\subseteq[k']}\widehat{g}(T)\sum_{x\in V^{B,b}}\chi'_{T}(x)\chi_{S}(x)\notag\\
&=\sum_{T\in\mathcal{T}(S)}\widehat{g}(T).
\end{align}
Next, to relate the squares of coefficients we want to have an upper bound
on $|\mathcal{T}(S)|$.
Denote $m=\sum_{\ell=1}^{k}|B_{\ell}|\cdot\mathbbm{1}\{|B_{\ell}|\geq 2\}$,
so that by assumption $m\leq \gamma n^{1/3}$. Note that the collections $\{\mathcal{T}(S)\}_{S\subseteq[n]}$ have the following property:
\begin{enumerate}
\item For $S_{1},S_{2}\subseteq[k]$ such that $S_{1}\neq S_{2}$, we have $\mathcal{T}(S_{1})\cap\mathcal{T}(S_{2})=\emptyset$.
\end{enumerate}

Next, let $Q\subseteq [k']$ be the set of all indices $j\in [k']$ such that $B'_{j}$ 
is contained in the set $\bigcup_{\ell \in [k]:\; |B_{\ell}|\geq 2}B_{\ell}$. Since the latter set has size $m$, we know that $Q$ has size at most $m$. By definition of $\mathcal{T}(S)$, it is not hard to see that a member $T\in\mathcal{T}(S)$ is uniquely determined by the set $Q\cap T$. 
Also, $|Q\cap T|\leq |T|-|S|+s$ where $s$ is the number of elements $\ell \in S$ such that $|B_{\ell}|\geq 2$. We thus get the following additional properties of $\mathcal{T}(S)$:
\begin{enumerate}
\setcounter{enumi}{1}
\item Fix $T\in\mathcal{T}(S)$. 
As each $\ell\in S$ contributes an odd number of elements to $T$ and each $\ell\not\in S$ contributes an even number of elements to $T$, we get that $|T|-|S|$ is an even integer. 
By the above, we get that it is in the range $\big[0,|Q|\big]\subseteq [0,m]$.
\item  For integer $j\in [0,m/2]$, the number of sets $T$ in $\mathcal{T}(S)$ with $|T|-|S|=2j$ is at most $\binom{|Q|}{2j+s}$, which we upper bound as $\binom{|Q|}{2j+s}\leq \binom{m}{2j}\cdot m^{|S|}$. Indeed, this is an upper bound on the number of choices for $Q\cap T$, and each such choice
determines a unique element in $\mathcal{T}(S)$.
\end{enumerate}

From this point onward in the proof, we will refer to the three properties mentioned listed above as the first, second and third properties of $\mathcal{T}(S)$.

By the second property of $\mathcal{T}(S)$, we know that if $|S|$ is odd then $|T|$ is odd for every $T\in\mathcal{T}(S)$, which implies $\widehat{g}(T)=0$ by the assumption on $g$, and thus $\widehat{h}(S)=0$ by~\eqref{eq:refine_fourier}. We are now ready to show that $h$ satisfies the desired Fourier weight bound on even degrees, and we split into cases.

\paragraph{Case 1: the low degree case.} Suppose that $|S|=2d$ where $0\leq d\leq \eta n^{1/3}$. Using~\eqref{eq:refine_fourier}, Cauchy-Schwarz and the third property of $\{\mathcal{T}(S)\}$ above we have:
\begin{align*}
\widehat{h}(S)^{2}&\leq
\left(\sum_{T\in\mathcal{T}(S)}n^{-(|T|-|S|)/3}\right)
\left(\sum_{T\in\mathcal{T}(S)}\widehat{g}(T)^{2}\cdot n^{(|T|- |S|)/3}\right)\\
&\leq 
\left(\sum_{j=0}^{\lfloor m/2\rfloor}\binom{m}{2j}m^{|S|}n^{-2j/3}\right)\left(\sum_{T\in\mathcal{T}(S)}\widehat{g}(T)^{2}\cdot n^{(|T|-|S|)/3}\right)\\
&\leq m^{|S|}(1+n^{-1/3})^{m}\sum_{T\in\mathcal{T}(S)}\widehat{g}(T)^{2}\cdot n^{(|T|-|S|)/3}\\
&\leq e^{\gamma}m^{|S|}\sum_{T\in\mathcal{T}(S)}\widehat{g}(T)^{2}\cdot n^{(|T|-|S|)/3},
\end{align*}
where we used $m\leq \gamma n^{1/3}$ in the last transition.
Using the first property of $\{\mathcal{T}(S)\}$, the decaying assumption of $g$, and then \Cref{prop:property-of-F}(2), we get
\begin{align*}
\sum_{|S|=2d}\widehat{h}(S)^{2}&\leq e^{\gamma}m^{2d}\sum_{j=0}^{\lfloor m/2\rfloor }n^{2j/3}\cdot\left(\sum_{|T|=2d+2j}\widehat{g}(T)^{2}\right)\\
&\leq e^{\gamma}m^{2d}\left\|g^{=2d}\right\|_{2}^{2}+e^{\gamma}m^{2d}\sum_{j=1}^{\lfloor m/2\rfloor }n^{2j/3}\cdot F\left(k',d+j,\eta n^{1/3}\right).\\
&\leq e^{\gamma}m^{2d}\left\|g^{=2d}\right\|_{2}^{2}+e^{\gamma}m^{2d}F\left(k',d,\eta n^{1/3}\right)\cdot\sum_{j=1}^{\lfloor m/2\rfloor }n^{2j/3}\left(\frac{2\eta n^{1/3}}{k'}\right)^{j} \\
&\leq e^{\gamma}m^{2d}\left\|g^{=2d}\right\|_{2}^{2}+e^{\gamma}m^{2d}\frac{4\eta}{1-4\eta}\cdot F\left(k',d,\eta n^{1/3}\right),
\end{align*}
where we used $k'\geq k\geq n/2$ in the last inequality. For the special case $d=0$ (recall that we defined $F(\cdot,0,\cdot)=1$) we have
$$\left\|h^{=0}\right\|_{2}^{2}\leq e^{\gamma}\cdot \frac{4\eta }{1-4\eta}\leq 8\eta + 2\gamma,$$
using $\gamma,\eta\in (0,\frac{1}{10})$. For $d>0$, using
$\left\|g^{=2d}\right\|_{2}^{2}\leq F\left(k',d,\eta n^{1/3}\right)$
we get
\begin{align*}
\left\|h^{=2d}\right\|_{2}^{2}
\leq \frac{e^{\gamma}m^{2d}}{1-4\eta}F\left(k',d,\eta n^{1/3}\right)
&\leq 2(\gamma n^{1/3})^{2d}\cdot F\left(k',d,\eta n^{1/3}\right)\\
&\leq F\left(k',d,2\eta\gamma^{2} n\right)&(\text{by \Cref{prop:property-of-F}(1)})\\
&\leq F\left(k',d,4\eta\gamma^{2}k\right) &(\text{using }k'\geq k\geq n/2)\\
&\leq F(k,d,4\eta\gamma^{2}k),
\end{align*}
where the last inequality is because $k'\geq k$.

\paragraph{Case 2: the high degree case.} 
Suppose that $|S|=2d$ where $d>\eta n^{1/3}$. Then by~\eqref{eq:refine_fourier}, Cauchy-Schwarz and the third property of $\{\mathcal{T}(S)\}$ above we have:
\begin{align*}
\widehat{h}(S)^{2}&\leq \left(\sum_{T\in\mathcal{T}(S)}\widehat{g}(T)^{2}\right)\left(\sum_{T\in\mathcal{T}(S)}1\right)
\leq \left(\sum_{T\in\mathcal{T}(S)}\widehat{g}(T)^{2}\right)\left(\sum_{j=1}^{\lfloor m/2\rfloor}\binom{m}{2j+s}\right)
\leq 2^{m-1}\sum_{T\in\mathcal{T}(S)}\widehat{g}(T)^{2}.
\end{align*}
So using the first property of $\{\mathcal{T}(S)\}$, we get:
\begin{align*}
\sum_{|S|=2d}\widehat{h}(S)^{2}&\leq 2^{m-1}\sum_{j=0}^{\lfloor m/2\rfloor}\sum_{|T|=2d+2j}\widehat{g}(T)^{2}
\\
&\leq 2^{m-1}\sum_{j=0}^{\lfloor m/2\rfloor}F\left(k',d+j,\gamma n^{1/3}\right)&(\text{by assumption on }g)\\
&\leq 2^{m-1}F\left(k',d,\gamma n^{1/3}\right)\sum_{j=0}^{\lfloor m/2\rfloor}2^{-j}&(\text{by \Cref{prop:property-of-F}(2)})\\
&\leq 2^{m}F\left(k',d,\gamma n^{1/3}\right)
\leq F\left(k',d,2^{\gamma/\min\{\gamma,\eta\}}\cdot\gamma n^{1/3}\right)&(\text{by \Cref{prop:property-of-F}(1)})\\
&\leq F\left(k,d,2^{\gamma/\min\{\gamma,\eta\}}\cdot\gamma n^{1/3}\right)&(\text{using }k'\geq k).
\end{align*}
Since $\gamma,\eta$ are constants, $k\geq n/2$ and $n$ is assumed to be sufficiently large, the bound obtained is at most $F(k,d,4\eta \gamma^{2}k)$.
\end{proof}

\subsection{Proof of \Cref{lem:norm}}
We are now ready to finish the proof of~\Cref{lem:norm}.
Fix a restriction $z'$, let $A\subseteq\Omega_{z'}$ be 
a $z'$-global set, let $z$ be a constraint that subsumes
$z'$ and let $h$ be as in the statement of the lemma.

\paragraph{Bringing the subspaces $L(z), L(z')$ to form:} 
let $G=([n],\supp(z))$ be the graph formed by the constraint $z$, 
and let $B_{1},B_{2},\dots,B_{k}$ be the connected components of $G$. For each $\ell\in[k]$, pick an arbitrary vertex $v_{\ell}\in B_{\ell}$ and set $b_{v_{\ell}}=0$. For any other vertex $u\in B_{\ell}$, pick the unique simple path $(u_{0},u_{1},\dots,u_{j})$ in $G$ such that $u_{0}=u$, $u_{j}=v_{\ell}$, and define
$$b_{u}=\sum_{i=1}^{j}z_{u_{i-1}u_{i}}.$$
By inspection, the affine subspace $V^{B,b}$ associated with the partition $B$ and string $b=(b_{1},\dots,b_{n})$ is exactly the affine subspace $L(z)$. 

Similarly, letting $G'=([n],\supp(z'))$ be the graph formed by the constraint $z'$, and letting $B'_{1},\dots,B'_{k'}$ be the connected components of $G'$, the graph $G'$ is a subgraph of $G$ and so $B'$ is a refinement of $B$. Using a similar reasoning to before we see that $V^{B',b}$ coincides with $L(z')$.

Since $z'$ is a restriction in the sense of \Cref{def:restrictions}, we have that each one of $B_{1}',\dots,B'_{k'}$ is a set of cardinality $1$ or $2$. Denoting $t=|\supp(z')|$, we assume without loss of generality that $B'_{1},\dots,B'_{n-2t}$ are the singletons among them, and $B'_{n-2t+1},\ldots,B'_{n-t}$ have size $2$. We also assume without loss of generality that $B'_{j}=\{j\}$ for every $j\in[n-2t]$.

\skipi
 Noting that $V^{B,b}\subseteq V^{B',b}$, we let $\iota\colon V^{B,b}\to V^{B',b}$ be the identity map. We recall the map $\id\colon V^{B',b}\colon \mathbb{F}_2^{k'}$ as in~\Cref{def:canonical_map}. Finally, consider the projection map $\pi:\mathbb{F}_{2}^{k'}\rightarrow \mathbb{F}_{2}^{n-2t}$ defined by $(x_{1},\dots,x_{k'})\mapsto (x_{1}+b_{1},\dots,x_{n-2t}+b_{n-2t})$. These maps give rise
 to the following diagram:
\begin{figure}[H]
\centering
\begin{tikzcd}
L(z)=V^{B,b} \arrow[r, "\iota"] \arrow[rrrd, "h", bend right = 10] & L(z') = V^{B',b}\arrow[r, "\id"] \arrow[rrd, "g"]& \mathbb{F}_{2}^{k'} \arrow[r, "\pi"] &\mathbb{F}_{2}^{n-2t} \arrow[d, "f"] \\
& & &\mathbb{R}
\end{tikzcd}
\caption{A commutative diagram of maps between sets}
\label{fig:commutative-diagram}
\end{figure}

\paragraph{Working in $\Omega_{z'}$:} 
Recall that we may identify $\Omega_{z'}$ with the space $\Omega^{[n-2t],\,\alpha n-t}$. This identification will 
enable us to think of $\Omega_{z'}$ as a matching space, and thus apply~\Cref{lem:global-decay} on it. 
Towards this end we adapt the definition of the linear spaces associated with restrictions and the probability distribution a set of edges induces. More precisely, for each $y\in \Omega_{z'}$ we define the affine subspace $L(y|z')\subseteq\mathbb{F}_{2}^{n-2t}$ by
    \[
    L(y|z'):=\bigcap_{\substack{\{u,v\}\subseteq [n-2t]\\ y_{uv}=1}}\left\{x\in\mathbb{F}_{2}^{n-2t}:x_{u}+ x_{v}=0\right\}\cap\bigcap_{\substack{\{u,v\}\subseteq [n-2t]\\ y_{uv}=-1}}\left\{x\in\mathbb{F}_{2}^{n-2t}:x_{u}+ x_{v}=1\right\},
    \]
and define a distribution $\Pr[\cdot |A,z']$ on $\mathbb{F}_{2}^{n-2t}$ by
\[\Pr[x|A,z']=\frac{1}{|A|}\sum_{y\in A}\frac{\mathbbm{1}\{x\in L(y|z')\}}{|L(y|z')|}.\]
Applying \Cref{lem:global-decay} to $A\subseteq \Omega_{z'}$, we see that the function $f:\mathbb{F}_{2}^{n-2t}\rightarrow\mathbb{R}$ defined by $f(x):=2^{n-2t}\Pr[x|A,z']-1$ is $(\eta n^{1/3}/2,0,2)$-decaying. Also, we note that
\begin{align*}
-1\leq f(x)\leq
\frac{2^{n-2t}}{|A|}\sum\limits_{y\in \Omega_{z'}}\frac{\mathbbm{1}\{x\in L(y|z')\}}{|L(y|z')|}
&=\frac{2^{\alpha n-t}}{|A|}
\sum\limits_{y\in \Omega_{z'}}\mathbbm{1}\{x\in L(y|z')\}\\
&=\frac{2^{\alpha n-t}\cdot\left|\Omega_{z'}\right|}{|A|}
\E_{y\in \Omega_{z'}}[\mathbbm{1}\{x\in L(y|z')\}].
\end{align*}
Sampling $y\in \Omega_{z'}$, for each $\{u,v\}\in \supp(y)\setminus \supp(z')$ the bit 
$y_{uv}$ is independent and is hence equal to 
$(-1)^{x_u+x_v}$ with probability $1/2$. It follows that the last expectation is equal to $2^{-(\alpha n - t)}$, and so $f(x)\leq |\Omega_{z'}|/|A|$. Summarizing, we get
\begin{equation}\label{eq:infty_norm_bd}
\norm{f}_{\infty}  
\leq
\frac{\left|\Omega_{z'}\right|}{|A|}
\leq 2^{\eta n^{1/3}}.
\end{equation}

\paragraph{Returning to the space $L(z')$:} 
we want to convert the information we have on $f$ to information about the function $h$ in the statement of the lemma.
Towards that end, we first note that $A\subseteq \Omega_{z'}$ directly induces a distribution $\Pr[\cdot|A]$ supported on $L(z')=V^{B',b}$, and we define $g(x):=2^{n-t}\Pr[x|A]-1$. Observe that 
\[
L(y)=(\pi\circ\id)^{-1}\left(L(y|z')\right),
\]
and as $\pi\circ \id$ is a linear map of rank $n-2t$ we get that $|L(y)|=2^{t}\left|L(y|z')\right|$. Thus, we get by definitions that 
\[
\Pr[x|A]=2^{-t}\Pr[\pi\circ\id(x)|A,z']
\]
for all $x\in V^{B',b}$, and hence $g(x)=f(\pi(\id(x)))$ 
(as also indicated by~\Cref{fig:commutative-diagram}). Considering the Fourier expansion, we get that for all $S\subseteq [k']$
\begin{align*}
\widehat{g}(S)
=2^{-n+t}\sum_{x\in V^{B',b}}g(x)\chi_{S}(\id(x))
&= 2^{-n+t}\sum_{\xi\in\mathbb{F}_{2}^{n-2t}}f(\xi)\left(\sum_{x\in \mathbb{F}_{2}^{n-t}}\chi_{S}(x)\cdot\mathbbm{1}\{\pi(x)=\xi\}\right)\\
&= 2^{-n+t}\sum_{\xi\in\mathbb{F}_{2}^{n-2t}}f(\xi)\left(\sum_{x\in \pi^{-1}(\xi)}
\chi_{S}(x)\right)\\
&=\begin{cases}
2^{-n+t}\sum_{\xi\in \mathbb{F}_{2}^{n-2t}}f(\xi)\cdot 2^{t}\chi_{S}(\xi+b) &\text{if }S\subseteq[n-2t]\\
0 &\text{if }S\not\subseteq[n-2t]
\end{cases}\\
&=\begin{cases}
\chi_{S}(b)\widehat{f}(S), &\text{if }S\subseteq [n-2t]\\
0, &\text{if }S\not\subseteq [n-2t].
\end{cases}
\end{align*}
Therefore, for every $d\geq 0$, we have $\left\|g^{=d}\right\|_{2}=\left\|f^{=d}\right\|_{2}$. 
Since $f$ is $(\eta n^{1/3}/2,0,2)$-decaying, we know that $f^{=d}\equiv 0$ for $d$ which is either $0$ or odd, 
yielding that $g^{=d}\equiv 0$ for such $d$'s. 
For even $d$'s, we get 
\[
\left\|g^{=2d}\right\|_{2}^{2}=
\left\|f^{=2d}\right\|_{2}^{2}\leq 2^{-d}F\Big(n-2t,d,\eta n^{1/3}/2\Big)\leq F\Big(n-t,d,\eta n^{1/3}/2\Big),
\]
where the last inequality follows by the definition of $F$.
This means $g$ is $(\eta n^{1/3}/2,0,1)$-decaying. 

To finish the proof, we note that $h:=g|_{V^{B,b}}$, so by~\Cref{lem:override-constraint} we get that $g$ is $(2\eta \gamma^{2} k,4\eta +2\gamma ,1)$-decaying. By~\eqref{eq:infty_norm_bd} we get that $\|h\|_{\infty}=\|f\|_{\infty}\leq 2^{\eta n^{1/3}}=2^{o(k)}$, so applying \Cref{lem:decay-to-K-norm} we conclude that 
\[\|h\|_{K}\leq\sqrt{4\eta\gamma^{2}K^{2}+4\eta + 2\gamma}+o(1).
\tag*{\qed}
\]
\section{Global Hypercontractivity in $\Omega$}\label{sec:global_hypercontractivity}

The goal of this section is to prove~\Cref{lem:global-decay}.

As discussed earlier, the proof proceeds by relating the Fourier coefficients of $h$ in the statement to Fourier-style coefficients of $1_A\colon \Omega \to\{0,1\}$, and then 
etablishing a type of level-$d$ inequality for functions 
over $\Omega$. Indeed, the bulk of the effort in our argument is devoted to the proof of a ``projected level-$d$ inequality'' for global functions over $\Omega$ (\Cref{thm:level-d-inequality}). Towards this end we first have to establish an appropriate global hypercontractive inequality, and we do so
using analogous results for product spaces from~\cite{KLM23}.


\subsection{Fourier Analytic Setup for Functions over $\Omega$}\label{subsec:Fourier-setup}
To facilitate Fourier-type analysis over $\Omega$, our first task is to introduce a collection of character functions on $\Omega^{U,m}$. The characters we define are indexed by ``partial matchings'' on the ground set $U$, that is, matchings of size smaller than or equal to $m$. To facilitate the definition, we first introduce the following notations.

\begin{definition}
For a ground set $U$ and an integer $d\geq 0$, we let $\calM_{U,d}$ denote the collection of all matchings over $U$ of size exactly $d$, and let $\calM_{U,\leq d}:=\bigcup_{s=0}^{d}\calM_{U,s}$.
\end{definition}
\begin{definition}\label{def:Psi}
For integers $n,m$ such that $n\geq 2m\geq 0$, we define $\Psi(n,m,0):=1$, and for $1\leq d\leq m$ we define inductively $\Psi(n,m,d):=m\binom{n}{2}^{-1}\cdot \Psi(n-2,m-1,d-1)$.
\end{definition}

It is easy to see that $\Psi(n,m,d)$ is equal to the probability that a fixed matching of size $d$ over a ground set of size $n$ is contained in a uniformly random matching of size $m$. We now define an associated collection of character:

\begin{definition}\label{def:characters}
For a matching $S\in\calM_{U,\leq m}$ and an element $y\in \Omega^{U,m}$, we define $y^{S}:= \prod_{\{u,v\}\in S}y_{uv}$. We then define the character function $\chi_{S}:\Omega^{U,m}\rightarrow\mathbb{R}$ by
$$\chi_{S}(y):=\Psi(|U|,m,|S|)^{-1/2}\cdot y^{S}.$$
\end{definition}

Note that the dimension of the inner product space $L^{2}(\Omega^{U,m})$ is much larger than the number of partial matchings on $U$. In particular, the character functions we defined do not form a basis for $L^{2}(\Omega^{U,m})$. Nevertheless, they will be sufficient for us, as Fourier coefficients coming from $h$ in the context of \Cref{lem:global-decay} will only be related to correlations with these character functions.

The next proposition shows the functions from~\Cref{def:characters} form an orthonormal set.
\begin{proposition}\label{prop:ortho_set}
For matchings $S,T\in\calM_{U,\leq m}$, we have $\langle \chi_{S},\chi_{T}\rangle=\mathbbm{1}\{S=T\}$.
\end{proposition}

\begin{proof}
First consider the case where $S\neq T$. If $S\cup T$ is not a matching, then $y^{S}\cdot y^{T}=0$ for all $y\in\Omega^{U,m}$, and therefore $\langle \chi_{S},\chi_{T}\rangle =0$. If $S\cup T$ is a matching, then conditioned on $S\cup T\subseteq \supp(y)$ (which is equivalent to $y^{S}\cdot y^{T}\neq 0$), the coordinates $\{y_{e}\}_{e\in S\cup T}$ are jointly uniformly distributed on $\{-1,1\}^{S\cup T}$. Therefore $S\neq T$ implies $\E_{y\in\Omega^{U,m}}[y^{S}\cdot y^{T}]=0$, and hence $\langle \chi_{S},\chi_{T}\rangle =0$.

Next, consider the case where $S=T$. Note that $\E_{y\in\Omega^{U,m}}[y^{S}\cdot y^{S}]=\Pr_{y\in\Omega^{U,m}}[S\subseteq \supp(y)]$, which equals the probability that a fixed matching of size $|S|$ is contained in a uniformly random matching of size $m$. Thus 
\[
\E_{y\in\Omega^{U,m}}\left[\chi_{S}(y)^{2}\right]= \Psi(|U|,m,|S|)^{-1}\cdot \E_{y\in\Omega^{U,m}}\left[y^{S}\cdot y^{S}\right]= \Psi(|U|,m,|S|)^{-1}\cdot \Psi(|U|,m,|S|)=1. \qedhere
\]
\end{proof}

\subsection{Discrete Derivatives and Derivated-based Globalness}\label{subsec:derivatives}
In this subsection we introduce the notion of discrete derivatives
for functions over $\Omega$, as well a related notion of globalness. These concepts are important in our subsequent 
argument, and more precisely in the derivation of the level-$d$ inequality from our global hypercontractive inequality.

For these purposes, we will want to further study ``restrictions'' as in \Cref{def:restrictions}. 
While in the context of~\Cref{def:restrictions}, restrictions on $\Omega^{U,m}$ are viewed as fixing a set of coordinates to certain values in $\{-1,1\}$, in this section, we are primarily interested in restrictions that require some coordinates to take values in $\{-1,1\}$ (i.e. not 0), without specifying which ones. 
Let $S$ be a matching on $U$ of size at most $m$, and consider the restricted domain $\{y\in \Omega^{U,m}:\supp(y)\supseteq S\}$. An element in this domain is determined by two choices: 
\begin{enumerate}
    \item Assigning labels to the edges in \( S \), which corresponds to selecting an element from \( \{-1,1\}^{S} \).
    \item Choosing the remaining labeled matching on \( U \setminus N(S) \), which corresponds to an element in \( \Omega^{U\setminus N(S),\,m-|S|} \).
\end{enumerate}
This leads to the following definition:

\begin{definition}\label{def:embedding}
For a matching $S\in \calM_{U,\leq m}$, there is a canonical embedding 
\[
\mathfrak{i}:\Omega^{U\setminus N(S),\,m-|S|}\times\{-1,1\}^{S}\hookrightarrow \Omega^{U,m}.
\]
This embedding proceeds by mapping a pair $(y,z)$ from the left hand side to the vector $\xi\in \Omega^{U,m}$ defined by $\xi_{uv}=y_{uv}$ for $\{u,v\}\subseteq U\setminus N(S)$, $\xi_{uv}=z_{uv}$ for $\{u,v\}\in S$, and $\xi_{uv}=0$ for all other pairs $\{u,v\}$. We will also use shorthand $\Omega^{U,m}_{\setminus S}$ to denote the space $\Omega^{U\setminus N(S),\,m-|S|}$.
\end{definition}

We now define derivative operators on $L^{2}(\Omega^{U,m})$:
\begin{definition}[Derivatives]\label{def:derivative}
Consider a function $f:\Omega^{U,m}\rightarrow\mathbb{R}$. For a matching $S\in \calM_{U,\leq m}$, we define a function $D_{S}[f]:\Omega^{U,m}_{\setminus S}\rightarrow\mathbb{R}$ by 
\[
D_{S}[f](y):=\E_{z\in\{-1,1\}^{S}}\left[z^{S}\cdot f(\mathfrak{i}(y,z))\right],
\]
where the embedding $\fraki:\Omega^{U,m}_{\setminus S}\times \{-1,1\}^{S}\rightarrow\Omega^{U,m}$ is as in \Cref{def:embedding}.
\end{definition}

A nice property of these derivative operators is that they are closed under composition.

\begin{lemma}\label{lem:derivatives-compose}
Suppose $S$ and $T$ are vertex disjoint matchings over $U$ with $|S\cup T|\leq m$. For $f:\Omega^{U,m}\rightarrow\mathbb{R}$ we have
\[
D_{S}D_{T}[f]=D_{S\cup T}[f].
\]
\end{lemma}
\begin{proof}
Recall from \Cref{def:embedding} that we have canonical embeddings
\[\fraki_{2}:\Omega^{U,m}_{\setminus (S\cup T)}\times\{-1,1\}^{S}\times\{-1,1\}^{T}\hookrightarrow \Omega^{U,m}\quad\text{ and }\quad \fraki_{1}:\Omega^{U,m}_{\setminus (S\cup T)}\times\{-1,1\}^{S}\hookrightarrow \Omega^{U,m}_{\setminus T}.\]
Thus, we get
\begin{align*}
D_{S}D_{T}[f](y)&=\E_{z_{(1)}\in \{-1,1\}^{S}}\left[z_{(1)}^{S}\cdot D_{T}[f]\Big(\fraki_{1}(y,z_{(1)})\Big)\right]\\
&=\E_{z_{(1)}\in \{-1,1\}^{S}}\left[z_{(1)}^{S}\cdot \E_{z_{(2)}\in\{-1,1\}^{T}}\left[z_{(2)}^{T}\cdot f\Big(\fraki_{2}(y,z_{(1)},z_{(2)})\Big)\right]\right]\\
&=\E_{z_{(1)}\in \{-1,1\}^{S}}\E_{z_{(2)}\in\{-1,1\}^{T}}\left[z_{(1)}^{S}z_{(2)}^{T}f\Big(\fraki_{2}(y,z_{(1)},z_{(2)})\Big)\right]\\
&=D_{S\cup T}[f](y).\qedhere
\end{align*}
\end{proof}

Following~\cite[Definition 4.4]{KLM23}, we define a notion of derivative-based globalness.

\begin{definition}\label{def:derivative-based-global}
Let $r,\lambda>0$ and $1\leq p<\infty$. For a function $f:\Omega^{U,m}\rightarrow\mathbb{R}$, we say it is $(r,\lambda,d)$-$L^{p}$-global if for every matching $S\in\calM_{U,\leq d}$, we have $\left\|D_{S}f\right\|_{p}\leq r^{|S|}\lambda$.
\end{definition}

\Cref{lem:derivatives-compose} then has the following important corollary.

\begin{corollary}\label{cor:globalness-of-derivative}
If $f:\Omega^{U,m}\rightarrow\mathbb{R}$ is $(r,\lambda,d)$-$L^{p}$-global, then for any matching $S\in\calM_{U,\leq d}$, the derivative $D_{S}[f]$ is $(r,r^{|S|}\lambda ,d-|S|)$-$L^{p}$-global.
\end{corollary}

\begin{proof}
For each matching $T\in \calM_{U\setminus N(S),\,\leq d-|S|}$, we know that $S\cup T\in \calM_{U,\leq d}$. So by the assumption that $f$ is $(r,\lambda,d)$-$L^{p}$-global, we have $\left\|D_{S\cup T}f\right\|_{p}\leq r^{|S|+|T|}\lambda$. By \Cref{lem:derivatives-compose} it follows that $\left\|D_{T}\left[D_{S}f\right]\right\|_{p}\leq r^{|T|}\cdot r^{|S|}\lambda$, as required.
\end{proof}

Thinking of the derivative with respect to $S$ 
as measuring the effect that the coordinates of $S$ 
have the the mass of the function, an intuitive way 
to think about the notion of derivative-based globaness
is that no small set of variables can bump the mass of 
$f$ too much beyond $\lambda$ (it is typical to think of $\lambda$ as $\Theta(\norm{f}_p)$. With this in mind, the following proposition (similar to~\cite[Lemma 4.9]{KLM23}) shows that globalness in the sense of~\Cref{def:global_set} implies discrete derivative based globalness as in~\Cref{def:derivative-based-global}.

\begin{proposition}\label{prop:connecting-two-globalness}
Suppose a subset $A\subseteq \Omega^{U,m}$ is a global set (in the sense of \Cref{def:global_set}). Let $\varphi:\Omega^{U,m}\rightarrow\{0,1\}$ be the indicator function of $A$. Then for every $1\leq p<\infty$, the function $\varphi$ is $(2^{1/p}, \|\varphi\|_{p}, m)$-$L^{p}$-global.
\end{proposition}

\begin{proof}
Consider an arbitrary matching $S\in\calM_{U,\leq m}$, and let $\fraki:\Omega^{U,m}_{\setminus S}\times \{-1,1\}^{S}\hookrightarrow \Omega^{U,m}$ be the embedding defined in \Cref{def:embedding}. For any fixed $z\in\{-1,1\}^{S}$, by \Cref{def:global_set}, the function $\varphi(\fraki(\cdot,z)):\Omega^{U,m}_{\setminus S}\rightarrow \{0,1\}$ is the indicator function of a set of size at most $2^{|S|}\cdot|A| \cdot |\Omega^{U,m}_{\setminus S}|/|\Omega^{U,m}|$. As $\varphi$ is Boolean valued we get 
$\left\|\varphi(\fraki(\cdot,z))\right\|_{p}^{p}\leq 2^{|S|}\cdot \left\|\varphi\right\|_{p}^{p}$, and so
\[
\left\|D_{S}[\varphi]\right\|_{p}=\left\|\E_{z\in\{-1,1\}^{S}}\left[z^{S}\cdot \varphi(\mathfrak{i}(\cdot,z))\right]\right\|_{p}\leq \E_{z\in\{-1,1\}^{S}}\left\|\varphi(\fraki(\cdot,z))\right\|_{p}\leq 2^{|S|/p}\cdot\left\|\varphi\right\|_{p}.
\qedhere
\]
\end{proof}

\subsection{Level-$d$ Projection}\label{subsec:projections}
To state and prove our level-$d$ inequality, we first define the projection operator on the space of degree $d$ functions 
spanned by the characters $\{\chi_S\}_{S\in \mathcal{M}_{U,d}}$.
\begin{definition}
Define $\calC^{U,m,d}=\lspan \left\{\chi_{S}:S\in\calM_{U,d}\right\}$ 
and define the operator $P_{\calC}^{=d}:L^{2}(\Omega^{U,m})\rightarrow\calC^{U,m,d}$ to be the orthogonal projection onto the subspace $\mathcal{C}^{U,m,d}$. 
\end{definition}

Using the fact that $\{\chi_{M}:M\in\calM_{U,d}\}$ forms an orthonormal basis of $\calC^{U,m,d}$ (see~\Cref{prop:ortho_set}), we have the following direct formula for projections.

\begin{proposition}\label{prop:formula-of-projection}
Given an integer $d\geq 0$, for each function $f:\Omega^{U,m}\rightarrow\mathbb{R}$ we have
$$P_{\calC}^{=d}[f](y):=\sum_{S\in\mathcal{M}_{U,d}}\langle f,\chi_{S}\rangle \cdot \chi_{S}(y).$$
\end{proposition}
\begin{proof}
    Follows immediately from~\Cref{prop:ortho_set}.
\end{proof}

Our next goal is to show that the projection operator (roughly) commutes with derivative operators. We first compute the derivative of character functions.  

\begin{proposition}\label{prop:derivative-for-pure-functions}
On the space $\Omega^{U,m}$, for all matchings $S,M\in \calM_{U,\leq m}$ we have
\[
D_{S}[\chi_{M}]=\begin{cases}
\Psi(|U|,m,|S|)^{-1/2}\cdot\chi_{M\setminus S} &\text{if }S\subseteq M,\\
0 &\text{if }S\not\subseteq M.
\end{cases}
\]
\end{proposition}

\begin{proof} We consider the following two cases respectively.

\paragraph{Case 1: $S\not\subseteq M$.} If $M\cup S$ is not a matching, then $\chi_{M}(\fraki(y,z))=0$ for all $y\in\Omega^{U,m}_{\setminus S}$ and $z\in\{-1,1\}^{S}$. If $M\cup S$ is a matching, then
$$\chi_{M}(\fraki(y,z))=\Psi(|U|,m,|M|)^{-1/2}\cdot y^{M\setminus S}\cdot z^{M\cap S}.$$
Using \Cref{def:derivative} and $S\neq M\cap S$, it follows that
\[
D_{S}[\chi_{M}](y)=\Psi(|U|,m,|M|)^{-1/2}\cdot \E_{z\in\{-1,1\}^{S}}\left[z^{S}\cdot y^{M\setminus S}\cdot z^{M\cap S}\right]=0.
\]

\paragraph{Case 2: $S\subseteq M$.} In this case we have
\begin{align*}
D_{S}[\chi_{M}](y)&=\Psi(|U|,m,|M|)^{-1/2}\cdot \E_{z\in\{-1,1\}^{S}}\left[z^{S}\cdot y^{M\setminus S}\cdot z^{M\cap S}\right]\\
&=
\Psi(|U|,m,|M|)^{-1/2}\cdot y^{M\setminus S}, \\
&=
\Psi(|U|,m,|M|)^{-1/2}\cdot \Psi\Big(|U\setminus N(S)|,m-|S|,|M\setminus S|\Big)^{1/2}\cdot\chi_{M\setminus S}(y)\\
&=
\Psi(|U|,m,|S|)^{-1/2}\cdot\chi_{M\setminus S}(y).
\end{align*}
In the above equation $\chi_{M}$ is the character on $\Omega^{U,m}$ while $\chi_{M\setminus S}$ is the character on $\Omega^{U,m}_{\setminus S}$.
\end{proof}

The next lemma shows that the operator $P_{\mathcal{C}}^{=d}$ 
commutes with derivative operators (up to taking into account the obvious change in degrees):
\begin{lemma}\label{lem:derivative-projection-commute}
For $f:\Omega^{U,m}\rightarrow\mathbb{R}$ and any matching $S\in\calM_{U,\leq d}$, we have $D_{S}P_{\calC}^{=d}[f]=P_{\calC}^{=d-|S|}D_{S}[f]$.
\end{lemma}
\begin{proof}
It follows from \Cref{prop:formula-of-projection,prop:derivative-for-pure-functions} that
\begin{equation}\label{eq:derivative-of-level-d}
D_{S}P_{\calC}^{=d}[f]=\Psi(|U|,m,|S|)^{-1/2}\sum_{\substack{M\in \mathcal{M}_{d}\\
M\supseteq S}}\langle f,\chi_{M}\rangle \cdot \chi_{M\setminus S}.
\end{equation}
Note that in the above equation $\chi_{M}$ is the character on $\Omega^{U,m}$ while $\chi_{M\setminus S}$ is the character on $\Omega^{U,m}_{\setminus S}$.

For any partial matching $T$ over $U\setminus N(S)$ of size $d-|S|$, we have
\begin{align*}
\left\langle D_{S}[f],\chi_{T}\right\rangle&= \Psi\Big(|U\setminus N(S)|,m-|S|,|T|\Big)^{-1/2}\cdot 
\E_{y\in \Omega^{U,m}_{\setminus S}}\left[\E_{z\in \{-1,1\}^{S}}\left[z^{S}\cdot f(\fraki(y,z))\right]\cdot y^{T}\right]\\
&=\Psi\Big(|U\setminus N(S)|,m-|S|,|T|\Big)^{-1/2}\cdot\E_{y\in \Omega^{U,m}_{\setminus S},z\in 
\{-1,1\}^{S}}\left[f(\fraki(y,z))\cdot y^{T}z^{S}\right]\\
&=\Psi\Big(|U\setminus N(S)|,m-|S|,|T|\Big)^{-1/2}\cdot \Psi(|U|,m,|S|)^{-1}\cdot\E_{\xi\in\Omega^{U,m}}\left[f(\xi)\cdot \xi^{S\cup T}\right]\\
&=\Psi\Big(|U\setminus N(S)|,m-|S|,|T|\Big)^{-1/2}\cdot \Psi(|U|,m,|S\cup T|)^{1/2}\cdot \Psi(|U|,m,|S|)^{-1}\cdot \langle f,\chi_{S\cup T}\rangle\\
&=\Psi(|U|,m,|S|)^{-1/2}\cdot \langle f,\chi_{S\cup T}\rangle.
\end{align*}
In the third transition above, we use the facts that $\xi^{S\cup T}$ is nonzero only if $\xi$ lies in the image of the embedding $\fraki:\Omega^{U,m}_{\setminus S}\times \{-1,1\}^{S}\hookrightarrow\Omega^{U,m}$, and that this image has size $\Psi(|U|,m,|S|)\cdot\left|\Omega^{U,m}\right|$. 
Comparing the above with \eqref{eq:derivative-of-level-d} shows that $D_{S}P_{\calC}^{=d}[f]=P_{\calC}^{=d-|S|}D_{S}[f]$.
\end{proof}

\subsection{The Hypercontractive Inequality}\label{subsec:hypercontractive}

In this section we establish our global hypercontractive inequality, which is later used in \Cref{subsec:level-d} 
to obtain the desired level-$d$ inequality with respect to the operator $P_{\mathcal{C}}^{=d}$ from the previous section.

Our proof uses a similar result in \cite{KLM23}, except that it is in the context of product spaces. To relate our setting and the product setting we introduce some setup. For a function $f\in \calC^{U,m,d}$, it will be useful to think of $f$ as a formal polynomial in $\binom{n}{2}$ variables. 

\begin{definition}\label{def:formal-polynomial}
For $f\in \calC^{U,m,d}$, we define its associated polynomial in the polynomial ring $\mathbb{R}\left[Y_{uv}:\{u,v\}\in\binom{U}{2}\right]$ to be
\[
\widetilde{f}(Y):=\sum_{M\in\calM_{U,d}}a_{M}\prod_{\{u,v\}\in M}Y_{uv},
\]
where $a_{M}=\langle f,\chi_{M}\rangle\cdot \Psi(|U|,m,d)^{-1/2}$. For each subset $S\subseteq \binom{U}{2}$ we define the formal derivative of $\widetilde{f}$ with respect to $S$ by
\[
\widetilde{D_{S}}\widetilde{f}(Y):=\sum_{\substack{M\in\calM_{U,d}\\ M\supseteq S}}a_{M}\prod_{\{u,v\}\in M\setminus S}Y_{uv}.
\]
\end{definition}

The following proposition follows easily from \Cref{def:characters} and~\Cref{prop:derivative-for-pure-functions} and shows that the generalization into formal polynomials and formal derivatives are backward compatible.

\begin{proposition}\label{prop:formal-compatibility}
For $f\in\calC^{U,m,d}$, \Cref{def:formal-polynomial} satisfies the following properties.
\begin{enumerate}
\item For $y\in\Omega^{U,m}$, we have $\widetilde{f}(y)=f(y)$.
\item For $S\subseteq \binom{U}{2}$, $\widetilde{D_{S}}\widetilde{f}$ is nonzero only if $S\in \calM_{U,\leq d}$.
\item For $S\in \calM_{U,\leq d}$, we have $\widetilde{D_{S}}\widetilde{f}=\widetilde{D_{S}f}$ as polynomials in the ring $\mathbb{R}\left[Y_{uv}:\{u,v\}\in\binom{U}{2}\right]$.
\end{enumerate}
\end{proposition}

Another observation useful for connecting the product space results with our setting is that although the uniform distribution on $\Omega^{U,m}$ is far from a product distribution, if it is projected onto a few coordinates, the projection is close to a product distribution. Concretely, this observation corresponds to the fact that the probability parameters $\Psi(n,m,d)$ defined in \Cref{def:Psi} grows approximately exponentially in $d$ when $d$ is small, as formalized below.

\begin{proposition}\label{prop:approximate-product}
Fix integers $n,m,d$ such that $n\geq 10 m$ and $m\geq 10(d+1)$. Let $p=\Psi(n,m,d)^{1/d}$.
\begin{enumerate}[label=(\arabic*)]
\item For $s\in\{0,1,\dots,d\}$ we have $p^{s}\leq \Psi(n,m,s)\leq (2p)^{s}$.
\item For $s\in\{d,d+1,\dots,m\}$ we have $\Psi(n,m,s)\leq p^{s}$.
\end{enumerate}
\end{proposition}
\begin{proof}
For $i\in\{0,1,\dots,m-1\}$, we have
$$\frac{\Psi(n-2i-2,m-i-1,1)}{\Psi(n-2i,m-i,1)}=\frac{\binom{n-2i}{2}}{\binom{n-2i-2}{2}}\cdot\frac{m-i-1}{m-i}\leq 
\left(1+\frac{2}{n-2i-3}\right)^2
\cdot\frac{m-i-1}{m-i}<1.$$
So $\Psi(n-2i,m-i,1)$ is decreasing in $i$. Furthermore, 
\[
\frac{\Psi(n,m,1)}{\Psi(n-2d,m-d,1)}=\frac{\binom{n-2d}{2}}{\binom{n}{2}}\cdot\frac{m}{m-d}\leq \frac{m}{m-d}\leq 2.
\]
Therefore, for $s\in\{0,1,\dots,d\}$ we have
\begin{align*}
\Psi(n,m,s)&=\prod_{i=0}^{s-1}\Psi(n-2i,m-i,1)\leq 2^{s}\cdot\Psi(n-2d,m-d,1)^{s}\\
&\leq 2^{s}\cdot\left(\prod_{i=0}^{d-1}\Psi(n-2i,m-i,1)\right)^{s/d}=2^{s}\cdot\Psi(n,m,d)^{s/d}=(2p)^{s},
\end{align*}
as well as
$$\Psi(n,m,s)=\prod_{i=0}^{s-1}\Psi(n-2i,m-i,1)\geq \left(\prod_{i=0}^{d}\Psi(n-2i,m-i,1)\right)^{s/d}=\Psi(n,m,d)^{s/d}=p^{s}.$$
For $s\geq d$ we have
\begin{align*}
\Psi(n,m,s)&=\Psi(n,m,d)\cdot\prod_{i=d}^{s-1}\Psi(n-2i,m-i,1)\leq \Psi(n,m,d)\cdot\Psi(n-2d,m-d,1)^{s-d}\\
&\leq \Psi(n,m,d)\cdot\left(\prod_{i=0}^{d-1}\Psi(n-2i,m-i,1)\right)^{(s-d)/d}=\Psi(n,m,d)^{s/d}=p^{s}.\qedhere
\end{align*}
\end{proof}

We need the following result, which is a direct consequence of~\cite[Theorem 4.1]{KLM23} adapted to our setting.

\begin{lemma}\label{lem:KLM-theorem}
Let $p\in(0,1)$ and let $z=\{z_{e}\}_{e\in \binom{U}{2}}$ be a set of mutually independent random variables, each following the distribution 
$$\Pr[z_{e}=-1]=\Pr[z_{e}=1]=\frac{p}{2},\text{ and }\Pr[z_{e}=0]=1-p.$$
Suppose $q$ is a positive integer and $\rho\in (0,\frac{1}{3\sqrt{2q}})$. For any $f\in\calC^{U,m,d}$ we have
\[
\E_{z}\left[\widetilde{f}(z)^{2q}\right]\leq \rho^{-2dq}\sum_{S\in\calM_{U,\leq d}}\beta^{2q|S|}(2q)^{-q|S|}\cdot p^{|S|}\cdot\E_{z}\left[\widetilde{D_{S}}\widetilde{f}(z)^{2}\right]^{q},
\]
where
$\beta:=\rho\sqrt{2q}\left(1+\frac{4(q-1)}{\ln(\rho^{-1}(2q)^{-1/2})}\right)$.
\end{lemma}

Utilizing \Cref{lem:KLM-theorem}, we obtain the following analogous hypercontractive inequality for our 
(non-product space) setting. In words, it says that if $f$ is discrete derivative global, then the $q$-norms of $f$ (for possibly large $q$) are bounded.
\begin{lemma}[Derivative-based hypercontractivity]\label{lem:derivative-based-inequality}
Suppose $|U|\geq 10m$ and $m\geq 10(d+1)$. Fix $r>0$ and integer $q\geq 1$. For $f\in \calC^{U,m,d}$ we have
\[
\left\|f\right\|_{2q}^{2q}\leq 2^{d}\rho^{-2dq}\left\|f\right\|_{2}^{2}\cdot\max_{S\in\calM_{U,\leq d}}\left(r^{-|S|}\left\|D_{S}f\right\|_{2}\right)^{2q-2},
\]
where
\begin{equation}\label{eq:rho-formula}
\rho:=\frac{1}{4\sqrt{2}}\min\left\{ q^{-1/2},q^{-1}r^{-\frac{q-1}{q}}\right\}.
\end{equation}
\end{lemma}
\begin{proof}
Let $z=\{z_{e}\}_{e\in \binom{U}{2}}$ be the random variables in \Cref{lem:KLM-theorem} with the parameter $p$ defined by  $p=\Psi(|U|,m,d)^{1/d}$, and write $f(y)=\sum_{M\in\calM_{U,d}}a_{M}y^{M}$. We define $g\in \calC^{U,m,d}$ to be $g(y):=\sum_{M\in\calM_{U,d}}|a_{M}|y^{M}$. 
Expanding, we get 
\[\E_{y\in\Omega^{U,m}}\left[f(y)^{2q}\right]=
\sum_{M_{1},\dots,M_{2q}\in \calM_{U,d}}
\E_{y\in\Omega^{U,m}}\left[\prod_{i=1}^{2q}y^{M_{i}}\right]\prod_{i=1}^{2q}a_{M_i}.\]
For each tuple of matchings $M_{1},\dots,M_{2q}\in\calM_{U,d}$, the expectation $\E_{y\in\Omega^{U,m}}\left[\prod_{i=1}^{2q}y^{M_{i}}\right]$ evaluates to either 0 (if some edge appears an odd number of times in these matchings) or some positive value. So replacing all $a_{i}$'s on the right hand side with their absolute values does not decrease the total value of the sum. We therefore have
\begin{align}
\E_{y\in\Omega^{U,m}}\left[f(y)^{2q}\right]&\leq \sum_{M_{1},\dots,M_{2q}\in \calM_{U,d}}
\E_{y\in\Omega^{U,m}}\left[\prod_{i=1}^{2q}y^{M_{i}}\right]\prod_{i=1}^{2q}|a_{M_i}|
\nonumber\\
&=\sum_{(M_{1},\dots,M_{2q})\in \textsf{Even}}
\Psi\left(|U|,m,\left|\bigcup_{i=1}^{2q}M_{i}\right|\right)\cdot\prod_{i=1}^{2q}|a_{M_i}|&(\text{by \Cref{def:Psi}})\nonumber\\
&\leq \sum_{(M_{1},\dots,M_{2q})\in \textsf{Even}}
p^{\left|\bigcup_{i=1}^{2q}M_{i}\right|}\cdot \prod_{i=1}^{2q}|a_{M_i}|&(\text{by \Cref{prop:approximate-product}(2)})\notag\\
&= \E_{z}\left[\widetilde{g}(z)^{2q}\right],\label{eq:assume-nonnegative-coefficients}
\end{align}
where $\textsf{Even}$ denotes the collection of tuples $(M_{1},\dots,M_{2q})\in \calM_{U,d}$ such that every edge appears an even number of times in these matchings.

Using \Cref{prop:formal-compatibility}, for $S\in\calM_{U,\leq d}$ we also have
\begin{align}
\E_{y\in\Omega^{U,m}_{\setminus S}}\left[D_{S}[f](y)^{2}\right]
=\E_{y\in\Omega^{U,m}_{\setminus S}}\left[\widetilde{D_{S}}\widetilde{f}(y)^{2}\right]
&=\E_{y\in\Omega^{U,m}_{\setminus S}}\left[\left(\sum_{M\in\calM_{U,d},\, M\supseteq S}a_{M}y^{M\setminus S}\right)^{2}\right] \nonumber\\
&=\sum_{M\in\calM_{U,d},\, M\supseteq S}a_{M}^{2}\cdot \Psi\Big(|U\setminus N(S)|,m-|S|,d-|S|\Big)\nonumber\\
&=\Psi(|U|,m,d)\cdot\Psi(|U|,m,|S|)^{-1}\cdot\sum_{M\in\calM_{U,d},\, M\supseteq S}a_{M}^{2}.\nonumber\\
&\geq p^{d}(2p)^{-|S|}\sum_{M\in\calM_{U,d},\, M\supseteq S}a_{M}^{2}\notag\\
&=2^{-|S|}\E_{z}\left[\widetilde{D_{S}}\widetilde{g}(z)^{2}\right],\label{eq:substitute-2} 
\end{align}
where the fifth transition uses~\Cref{prop:approximate-product}(1). Plugging \eqref{eq:assume-nonnegative-coefficients} and \eqref{eq:substitute-2} into \Cref{lem:KLM-theorem} yields
\begin{equation}\label{eq:derivative-based-untidied}
\left\|f\right\|_{2q}^{2q}\leq \rho^{-2dq}\sum_{S\in\calM_{U,\leq d}}\beta^{2q|S|}(2q)^{-q|S|}\cdot p^{|S|}\cdot2^{q|S|}\left\|D_{S}f\right\|_{2}^{2q}
= \rho^{-2dq}\sum_{S\in\calM_{U,\leq d}}\beta^{2q|S|}q^{-q|S|}\cdot p^{|S|}\cdot\left\|D_{S}f\right\|_{2}^{2q},
\end{equation}
where $\beta:=\rho\sqrt{2q}\left(1+\frac{4(q-1)}{\ln(\rho^{-1}(2q)^{-1/2})}\right)$ and $\rho$ is given by \eqref{eq:rho-formula}. 
Using $\rho^{-1}(2q)^{-1/2}\geq 4$ we get $\beta\leq \rho\sqrt{2q}\cdot 4q$. Now the other upper bound $\rho\leq \frac{1}{4\sqrt{2}}q^{-1}r^{-\frac{q-1}{q}}$ yields $\beta^{2q}q^{-q}r^{2q-2}\leq 1$. Plugging back into \eqref{eq:derivative-based-untidied}, we get
\begin{align}
\left\|f\right\|_{2q}^{2q}&\leq \rho^{-2dq}\sum_{S\in\calM_{U,\leq d}} p^{|S|}\cdot\left\|D_{S}f\right\|_{2}^{2}\cdot\left(r^{-|S|}\left\|D_{S}f\right\|_{2}\right)^{2q-2}\nonumber\\
&\leq \rho^{-2dq}\left(\sum_{S\in\calM_{U,\leq d}} p^{|S|}\cdot\left\|D_{S}f\right\|_{2}^{2}\right)\cdot\max_{S\in\calM_{U,\leq d}}\left(r^{-|S|}\left\|D_{S}f\right\|_{2}\right)^{2q-2}.\label{eq:derivative-based-tidied}
\end{align}
Now using \Cref{prop:approximate-product}(1) again, we have
\begin{align*}
\sum_{S\in\calM_{U,\leq d}} p^{|S|}\cdot\left\|D_{S}f\right\|_{2}^{2}&\leq \sum_{S\in\calM_{U,\leq d}} \Psi(|U|,m,|S|)\cdot\left\|D_{S}f\right\|_{2}^{2}\\
&= \sum_{S\in\calM_{U,\leq d}}\Psi(|U|,m,|S|)\cdot\E_{y\in\Omega^{U,m}_{\setminus S}}\left[\left(\sum_{M\in\calM_{U,d},\,M\supseteq S}a_{M}y^{M\setminus S}\right)^{2}\right]\\
&=\sum_{S\in\calM_{U,\leq d}}\Psi(|U|,m,|S|)\Psi\Big(|U\setminus N(S)|,m-|S|,d-|S|\Big)
\sum_{\substack{M\in\calM_{U,d}\\ M\supseteq S}}a_{M}^{2}\\
&=\sum_{S\in\calM_{U,\leq d}}\Psi(|U|,m,d)\sum_{\substack{M\in\calM_{U,d}\\ M\supseteq S}}a_{M}^{2}=2^{d}\sum_{M\in\calM_{U,d}}\Psi(|U|,m,d)\cdot a_{M}^{2}\\
&=2^{d}\E_{y\in\Omega^{U,m}}\left[\left(\sum_{M\in\calM_{U,d}}a_{M}y^{M}\right)^{2}\right]=2^{d}\left\|f\right\|_{2}^{2}.
\end{align*}
Plugging the above into \eqref{eq:derivative-based-tidied} yields the conclusion.
\end{proof}

\subsection{The Level-$d$ Inequality}\label{subsec:level-d}
We next run an induction argument along the lines of~\cite[Theorem 5.2]{KLM23} to prove the projected level-$d$ inequality, formally stated below. To get some sense of the 
statement and the argument, the reader should keep in mind the setting that 
$f$ is $\{0,1\}$-valued, and $\lambda_1 = \|f\|_1$, $\lambda_2 = \|f\|_2$, so that $\lambda_2^2 = \lambda_1$, and in particular $\lambda_1$ is much smaller than $\lambda_2$. Indeed, level-$d$ inequalities are typically stated and used for Boolean functions. However, the argument presented here proceeds by induction on the degree parameter $d$, and we appeal for the inductive hypothesis on the discrete derivatives of $f$. Note that if $f$ is Boolean, then $D_S[f]$ need not be Boolean. However, it is easy to note that the values of $D_S[f]$ are integer multiples of $2^{-|S|}$, so its values are not completely arbitrary, and in particular there is still a large gap between the $1$-norm and $2$-norm of $D_S[f]$. Thus,
the parameter that replaces Booleanity in the below statement is essentially that the ratio $\lambda_2/\lambda_1$ is large.
\begin{theorem}[Projected level-$d$ inequality]\label{thm:level-d-inequality}
Suppose $|U|\geq 10m$ and $m\geq 10(d+1)$. Suppose $f:\Omega^{U,m}\rightarrow\mathbb{R}$ is both $(r,\lambda_{1},d)$-$L^{1}$-global and $(r,\lambda_{2},d)$-$L^{2}$-global, where $d\leq \log(\lambda_{2}/\lambda_{1})$ and $r\geq 1$. Then
\[\left\|P_{\calC}^{=d}f\right\|_{2}^{2}\leq \lambda_{1}^{2}\left(\frac{10^{5}r^{2}\log(\lambda_{2}/\lambda_{1})}{d}\right)^{d}.\]
\end{theorem}

\begin{proof}

The conclusion in the case $d=0$ is simply $\E_{y}[f(y)]^{2}\leq \lambda_{1}^{2}$, which holds by the $L^1$-globalness assumption. We proceed by an induction on $d$. 
Towards this end, fix $d\geq 1$ and assume that the statement holds for all $d'<d$.

Fix $S\neq \emptyset$, so that by \Cref{cor:globalness-of-derivative} we know that $D_{S}f$ is both $(r,r^{|S|}\lambda_{1},d-|S|)$-$L^{1}$-global and $(r,r^{|S|}\lambda_{2},d-|S|)$-$L^{2}$-global. 
Our first goal 
will be to show that $P_{\calC}^{=d}[f]$ has discrete derivatives with small norms, and towards this end we use the induction hypothesis. 
Since $|U|-2|S|\geq 10(m-|S|)$ and $m-|S|\geq 10(d-|S|+1)$, we can apply the induction hypothesis on $D_{S}[f]:\Omega^{U,m}_{\setminus S}\rightarrow \mathbb{R}$. Combining with \Cref{lem:derivative-projection-commute}, we get
\begin{align}
\left\|D_{S}P_{\calC}^{=d}[f]\right\|_{2}^{2}&= \left\|P_{\calC}^{=d-|S|}D_{S}[f]\right\|_{2}^{2}\\
&\leq r^{2|S|}\lambda_{1}^{2}\left(\frac{10^{5}r^{2}\log(\lambda_{2}/\lambda_{1})}{d-|S|}\right)^{d-|S|}\nonumber\\
&= 10^{5(d-|S|)}\lambda_{1}^{2}r^{2d}\log^{d-|S|}(\lambda_{2}/\lambda_{1})d^{-(d-|S|)}\left(1+\frac{|S|}{d-|S|}\right)^{d-|S|}\nonumber\\
&\leq 10^{5(d-|S|)}\lambda_{1}^{2}r^{2d}\log^{d-|S|}(\lambda_{2}/\lambda_{1})d^{-(d-|S|)}\cdot 10^{5|S|}\nonumber\\
&= \lambda_{1}^{2}\left(\frac{10^{5}r^{2}\log(\lambda_{2}/\lambda_{1})}{d}\right)^{d}\left(\frac{\sqrt{d}}{\log^{1/2}(\lambda_{2}/\lambda_{1})}\right)^{2|S|}\notag\\
&=(r')^{2|S|}(\lambda')^{2},\label{eq:globalness-of-level-d}
 \end{align}
 where we let 
\[\lambda'=\lambda_{1}\left(\frac{10^{5}r^{2}\log(\lambda_{2}/\lambda_{1})}{d}\right)^{d/2}\quad\text{ and }\quad
r'=\frac{\sqrt{d}}{\log^{1/2}(\lambda_{2}/\lambda_{1})}.
\]
We intend to apply~\Cref{lem:derivative-based-inequality}, and for that we pick
\[q=\left\lfloor\frac{4\log (\lambda_{2}/\lambda_{1})}{d}\right\rfloor\quad\text{ and }\quad\rho=  \frac{1}{4\sqrt{2}}\min\left\{ q^{-1/2},q^{-1}(r')^{-\frac{q-1}{q}}\right\}.\]
This choice of parameters ensure that $\rho^{-2}\leq 10^{3}q$, and thus
\begin{align}
2^{d}\rho^{-2dq}\lambda_{1}^{2q}\left(\lambda_{2}/\lambda_{1}\right)^{2} &\leq \lambda_{1}^{2q}\left(2\rho^{-2}(\lambda_{2}/\lambda_{1})^{2/(dq)}\right)^{dq}\leq \lambda_{1}^{2q}(10^{4}q)^{dq}\nonumber\\
&\leq \lambda_{1}^{2q}\left(\frac{10^{5}r^{2}\log(\lambda_{2}/\lambda_{1})} {d}\right)^{dq}=(\lambda')^{2q}.\label{eq:parameter-mess-to-lambda}
\end{align}
Since $P_{\calC}^{=d}$ is an orthogonal projection, we have
\begin{align*}
\left\|P_{\calC}^{=d}f\right\|_{2}^{4q}
=\left\langle f,P_{\calC}^{=d}f\right\rangle^{2q}
&\leq \left\|P_{\calC}^{=d}f\right\|_{2q}^{2q}\cdot \|f\|_{2q/(2q-1)}^{2q}\\
&\leq \left\|P_{\calC}^{=d}f\right\|_{2q}^{2q}\cdot \|f\|_{1}^{2q-2}\cdot \|f\|_{2}^{2}\\
&\leq 2^{d}\rho^{-2dq}\lambda_{1}^{2q-2}\lambda_{2}^{2}\left\|P_{\calC}^{=d}f\right\|_{2}^{2}\cdot \max_{S\in\calM_{U,\leq d}}\left((r')^{-|S|}\left\|D_{S}P_{\calC}^{=d}f\right\|_{2}\right)^{2q-2}\\
&\leq (\lambda')^{2q}\cdot\left\|P_{\calC}^{=d}f\right\|_{2}^{2}\max\left(\left\|P_{\calC}^{=d}f\right\|_{2}^{2q-2},(\lambda')^{2q-2}\right),
\end{align*}
where the second and third transitions are by H\"{o}lder's inequality, the fourth transition is by~\Cref{lem:derivative-based-inequality}, and the last transition is by 
\eqref{eq:parameter-mess-to-lambda} and \eqref{eq:globalness-of-level-d}. It follows that $\left\|P_{\calC}^{=d}f\right\|_{2}\leq \lambda'$, as desired.
\end{proof}

\subsection{Proof of \Cref{lem:global-decay}}\label{subsec:proof-of-decay}
We now prove \Cref{lem:global-decay}, restated below.
\globalimpliesdecay*
\begin{proof}
    By definition of the function $f$, we know that $\widehat{f}(\emptyset)=0$. For any non-empty subset $S\subseteq U$, we know from \Cref{def:conditional-probability} that
    \begin{align}
        \widehat{f}(S) &= \frac{1}{2^{|U|}}\sum_{x\in\mathbb{F}_{2}^{U}}f(x)\cdot(-1)^{\sum_{i\in S}x_{i}}\nonumber\\
        &=\frac{1}{|A|}
        \sum_{y\in A}\left(\frac{1}{|L(y)|}\sum_{x\in\F_{2}^{U}}\mathbbm{1}\{x\in L(y)\}\cdot (-1)^{\sum_{i\in S}x_{i}}\right)\nonumber\\
        &=\frac{1}{|A|}\sum_{y\in A}\left(\mathbbm{1}\left\{L(y)\subseteq\left\{x\in\F_{2}^{U}:\sum_{i\in S}x_{i}=0\right\}\right\}-\mathbbm{1}\left\{L(y)\subseteq\left\{x\in\F_{2}^{U}:\sum_{i\in S}x_{i}=1\right\}\right\}\right)\label{eq:difficult-to-express}.
    \end{align}
    The last transition holds because the function 
    $(-1)^{\sum_{i\in S}x_{i}}$ is constant in the 
    events appearing in the last expression, and else
    it has average $0$ inside $L(y)$.
    
    Let $\mathcal{M}(S)$ be the collection of perfect matchings of the vertices in $S$ (if $|S|$ is odd then $\mathcal{M}(S) = \emptyset$). For a given labeled matching $y\in\Omega$, observe that $L(y)\subseteq \{x\in \mathbb{F}_{2}^{U}:\sum_{i\in S}x_{i}=0\}$ if and only if there exists some $M\in\calM(S)$ such that $\prod_{\{u,v\}\in M}y_{uv}=1$. Similarly, $L(y)\subseteq \{x\in \mathbb{F}_{2}^{U}:\sum_{i\in S}x_{i}=1\}$ if and only if there exists some $M\in\calM(S)$ such that $\prod_{\{u,v\}\in M}y_{uv}=-1$. Furthermore, for a given $y\in \Omega$ there can be at most one such $M\in \calM(S)$ such that $\prod_{\{u,v\}\in M}y_{uv}\neq 0$. We thus have
    \begin{align}
    &\quad\mathbbm{1}\left\{L(y)\subseteq\left\{x\in\F_{2}^{U}:\sum_{i\in S}x_{i}=0\right\}\right\}-\mathbbm{1}\left\{L(y)\subseteq\left\{x\in\F_{2}^{U}:\sum_{i\in S}x_{i}=1\right\}\right\} \nonumber\\
    &=\sum_{M\in\mathcal{M}(S)}\prod_{\{u,v\}\in M} y_{uv}=\begin{cases}
    0 &\text{if }|S|\text{ is odd},\\
    \Psi(|U|,m,|S|/2)^{1/2}\cdot \sum_{M\in\calM(S)}\chi_{M}(y)&\text{if }|S|\text{ is even}.
    \end{cases}
    \label{eq:clean}
    \end{align}
    The second transition above is by the definition of the characters.
    
    Plugging into \eqref{eq:difficult-to-express}, we see that $\widehat{f}(S)=0$ for all $S\subseteq [n]$ of odd size, as required for $f$ to be $(w,0,2)$-decaying. In the following, we only focuse on sets of even size and Fourier weights on even levels. 
    
    Let $\varphi:\Omega^{U,m}\rightarrow\{0,1\}$ be the indicator function of $A$. When $|S|=2d$ for some $d\geq 1$, again by plugging \eqref{eq:clean} into \eqref{eq:difficult-to-express} we get 
    \[
    \widehat{f}(S)=\frac{1}{|A|}\Psi(|U|,m,d)^{1/2}\cdot \sum_{y\in A}\sum_{M\in\calM(S)}\chi_{M}(y)
    =\frac{\left|\Omega^{U,m}\right|}{|A|}\cdot\Psi(|U|,m,d)^{1/2}\cdot \sum_{M\in\calM(S)}\left<\varphi,\chi_{M}\right>.
    \]
    Note that since $|S|=2d$, the collection $\calM(S)$ has size $(2d-1)!!$, the number of perfect matchings in a complete graph of $2d$ vertices. We then apply Cauchy-Schwarz to the above equation and get
    \begin{align*}
    \widehat{f}(S)^{2}&\leq \frac{\left|\Omega^{U,m}\right|^{2}}{|A|^{2}}\cdot\Psi(|U|,m,d)\cdot |\calM(S)|\cdot\sum_{M\in\calM(S)}\langle \varphi,\chi_{M}\rangle^{2}\\
    &=\frac{1}{\|\varphi\|_{1}^{2}}\left(\prod_{i=0}^{d-1}\frac{m-i}{\binom{|U|-2i}{2}}\right)\cdot(2d-1)!!\cdot\sum_{M\in\calM(S)}\langle \varphi,\chi_{M}\rangle^{2}\\
    &=\frac{1}{\|\varphi\|_{1}^{2}}\binom{m}{d}\binom{|U|}{2d}^{-1}\cdot\sum_{M\in\calM(S)}\langle \varphi,\chi_{M}\rangle^{2}.
    \end{align*}
    For every matching $M\in\calM_{U,d}$, there exists exactly one set $S\subseteq [n]$ of size $2d$ such that $M\in\calM(S)$. Therefore we deduce that
    \begin{align}
    \left\|f^{=2d}\right\|_{2}^{2}&=\sum_{\substack{S\subseteq [n]\\ |S|=2d}}\widehat{f}(S)^{2}\leq \frac{1}{\|\varphi\|_{1}^{2}}\binom{m}{d}\binom{|U|}{2d}^{-1}\cdot\sum_{M\in\calM_{U,d}}\langle\varphi,\chi_{M}\rangle^{2}\nonumber\\
    &=\frac{1}{\|\varphi\|_{1}^{2}}\binom{m}{d}\binom{|U|}{2d}^{-1}\cdot\left\|P_{\calC}^{=d}\varphi\right\|_{2}^{2}.\label{eq:varphi-to-f}
    \end{align}
    Since $A$ is a global set in $\Omega^{U,m}$ of size $2^{-w}\cdot\left|\Omega^{U,m}\right|$, we know from \Cref{prop:connecting-two-globalness} that its indicator function $\varphi$ is both $(2,2^{-w},d)$-$L^{1}$-global and $(2,2^{-w/2},d)$-$L^{2}$-global. If $d\leq w/2$, we can apply \Cref{thm:level-d-inequality} to the right hand side of \eqref{eq:varphi-to-f} and get 
    $$\left\|f^{=2d}\right\|_{2}^{2}\leq \binom{m}{d}\binom{|U|}{2d}^{-1}\cdot\left(\frac{10^{5}\cdot 4\cdot (w/2)}{d}\right)^{d}\leq \left(\frac{3m}{d}\right)^{d}\left(\frac{2d}{|U|}\right)^{2d}\left(\frac{2\cdot10^{5}w}{d}\right)^{d}\leq \left(\frac{w/2}{2|U|}\right)^{d},$$
    since $m\leq 10^{-7}|U|$. If $d>w/2$, we note that $\left\|P_{\calC}^{=d}\varphi\right\|_{2}^{2}\leq \left\|\varphi\right\|_{2}^{2}$. Then from \eqref{eq:varphi-to-f} we have
    $$
    \left\|f^{=2d}\right\|_{2}^{2}\leq \frac{\left\|\varphi\right\|_{2}^{2}}{\left\|\varphi\right\|_{1}^{2}}\cdot \binom{m}{d}\binom{|U|}{2d}^{-1}\leq 
    2^{w}\left(\frac{3m}{d}\right)^{d}\left(\frac{2d}{|U|}\right)^{2d}\leq 2^{w}\cdot\left(\frac{d}{8|U|}\right)^{d},$$
    since $m\leq 10^{-7}|U|$. Combining the above two equations, we conclude for any $d\geq 1$ that $\left\|f^{=2d}\right\|_{2}^{2}\leq 2^{-d}F(|U|,d,w/2)$, and we thus conclude the proof.
    \end{proof}

\bibliographystyle{alpha}
\bibliography{ref.bib}

\newcommand{\etalchar}[1]{$^{#1}$}
\begin{thebibliography}{KKMO07}

\bibitem[AKSY20]{AKSY20}
Sepehr Assadi, Gillat Kol, Raghuvansh~R. Saxena, and Huacheng Yu.
\newblock Multi-pass graph streaming lower bounds for cycle counting, max-cut, matching size, and other problems.
\newblock In {\em 2020 IEEE 61st Annual Symposium on Foundations of Computer Science (FOCS)}, pages 354--364, 2020.

\bibitem[AMS96]{AMS96}
Noga Alon, Yossi Matias, and Mario Szegedy.
\newblock The space complexity of approximating the frequency moments.
\newblock In {\em Proceedings of the Twenty-Eighth Annual ACM Symposium on Theory of Computing}, STOC '96, page 20–29, New York, NY, USA, 1996. Association for Computing Machinery.

\bibitem[AN21]{AN21}
Sepehr Assadi and Vishvajeet N.
\newblock Graph streaming lower bounds for parameter estimation and property testing via a streaming xor lemma.
\newblock In {\em Proceedings of the 53rd Annual ACM SIGACT Symposium on Theory of Computing}, STOC 2021, page 612–625, New York, NY, USA, 2021. Association for Computing Machinery.

\bibitem[Ass23]{Assadi}
Sepehr Assadi.
\newblock Recent advances in multi-pass graph streaming lower bounds.
\newblock {\em SIGACT News}, 54(3):48–75, September 2023.

\bibitem[CGS{\etalchar{+}}22]{chou2022linear}
Chi-Ning Chou, Alexander Golovnev, Madhu Sudan, Ameya Velingker, and Santhoshini Velusamy.
\newblock Linear space streaming lower bounds for approximating {CSP}s.
\newblock In {\em Proceedings of the 54th Annual ACM SIGACT Symposium on Theory of Computing}, pages 275--288, 2022.

\bibitem[CGSV24]{chou2024sketching}
Chi-Ning Chou, Alexander Golovnev, Madhu Sudan, and Santhoshini Velusamy.
\newblock Sketching approximability of all finite {CSP}s.
\newblock {\em Journal of the ACM}, 71(2):1--74, 2024.

\bibitem[CGV20]{chou2020optimal}
Chi-Ning Chou, Alexander Golovnev, and Santhoshini Velusamy.
\newblock Optimal streaming approximations for all boolean max-2csps and max-ksat.
\newblock In {\em 2020 IEEE 61st Annual Symposium on Foundations of Computer Science (FOCS)}, pages 330--341. IEEE, 2020.

\bibitem[CKP{\etalchar{+}}23]{CKP+23}
Lijie Chen, Gillat Kol, Dmitry Paramonov, Raghuvansh~R Saxena, Zhao Song, and Huacheng Yu.
\newblock Towards multi-pass streaming lower bounds for optimal approximation of max-cut.
\newblock In {\em Proceedings of the 2023 Annual ACM-SIAM Symposium on Discrete Algorithms (SODA)}, pages 878--924. SIAM, 2023.

\bibitem[DKK{\etalchar{+}}18]{DKKMS1}
Irit Dinur, Subhash Khot, Guy Kindler, Dor Minzer, and Muli Safra.
\newblock Towards a proof of the 2-to-1 {Games Conjecture}?
\newblock In {\em {STOC} 2018}, pages 376--389, 2018.

\bibitem[DKK{\etalchar{+}}21]{DKKMS2}
Irit Dinur, Subhash Khot, Guy Kindler, Dor Minzer, and Muli Safra.
\newblock On non-optimally expanding sets in {G}rassmann graphs.
\newblock {\em Israel journal of mathematics}, 243(1):377--420, 2021.

\bibitem[EKL22]{EKLapp}
David Ellis, Guy Kindler, and Noam Lifshitz.
\newblock Forbidden intersection problems for families of linear maps.
\newblock {\em arXiv preprint arXiv:2208.04674}, 2022.

\bibitem[GKK{\etalchar{+}}07]{gavinsky2007exponential}
Dmitry Gavinsky, Julia Kempe, Iordanis Kerenidis, Ran Raz, and Ronald De~Wolf.
\newblock Exponential separations for one-way quantum communication complexity, with applications to cryptography.
\newblock In {\em Proceedings of the thirty-ninth annual ACM symposium on Theory of computing}, pages 516--525, 2007.

\bibitem[GPW17]{GPW17}
Mika Göös, Toniann Pitassi, and Thomas Watson.
\newblock Query-to-communication lifting for {BPP}.
\newblock In {\em 2017 IEEE 58th Annual Symposium on Foundations of Computer Science (FOCS)}, pages 132--143, 2017.

\bibitem[GS24]{green2024improved}
Ben Green and Mehtaab Sawhney.
\newblock Improved bounds for the {F}urstenberg-{S}arkozy theorem.
\newblock {\em arXiv preprint arXiv:2411.17448}, 2024.

\bibitem[GVV17]{guruswami2017streaming}
Venkatesan Guruswami, Ameya Velingker, and Santhoshini Velusamy.
\newblock Streaming complexity of approximating max 2csp and max acyclic subgraph.
\newblock In {\em Approximation, Randomization, and Combinatorial Optimization. Algorithms and Techniques (APPROX/RANDOM 2017)}, pages 8--1. Schloss Dagstuhl--Leibniz-Zentrum f{\"u}r Informatik, 2017.

\bibitem[GW95]{GW95}
Michel~X. Goemans and David~P. Williamson.
\newblock Improved approximation algorithms for maximum cut and satisfiability problems using semidefinite programming.
\newblock {\em J. ACM}, 42(6):1115–1145, November 1995.

\bibitem[HSV24]{hwang2024oblivious}
Samuel Hwang, Noah~G Singer, and Santhoshini Velusamy.
\newblock Oblivious algorithms for maximum directed cut: New upper and lower bounds.
\newblock {\em arXiv preprint arXiv:2411.12976}, 2024.

\bibitem[Kho02]{Khot02}
Subhash Khot.
\newblock On the power of unique 2-prover 1-round games.
\newblock In {\em Proceedings of the 17th Annual {IEEE} Conference on Computational Complexity, Montr{\'{e}}al, Qu{\'{e}}bec, Canada, May 21-24, 2002}, page~25. {IEEE} Computer Society, 2002.

\bibitem[KK19]{KK19}
Michael Kapralov and Dmitry Krachun.
\newblock An optimal space lower bound for approximating max-cut.
\newblock In {\em Proceedings of the 51st Annual ACM SIGACT Symposium on Theory of Computing}, STOC 2019, page 277–288, New York, NY, USA, 2019. Association for Computing Machinery.

\bibitem[KKMO07]{KKMO07}
Subhash Khot, Guy Kindler, Elchanan Mossel, and Ryan O’Donnell.
\newblock Optimal inapproximability results for max-cut and other 2-variable {CSP}s?
\newblock {\em SIAM J. Comput.}, 37(1):319–357, April 2007.

\bibitem[KKS15]{KKS15}
Michael Kapralov, Sanjeev Khanna, and Madhu Sudan.
\newblock Streaming lower bounds for approximating max-cut.
\newblock In {\em Proceedings of the Twenty-Sixth Annual ACM-SIAM Symposium on Discrete Algorithms}, SODA '15, page 1263–1282, USA, 2015. Society for Industrial and Applied Mathematics.

\bibitem[KLLM23]{keevash2023forbidden}
Peter Keevash, Noam Lifshitz, Eoin Long, and Dor Minzer.
\newblock Forbidden intersections for codes.
\newblock {\em Journal of the London Mathematical Society}, 108(5):2037--2083, 2023.

\bibitem[KLLM24]{KLLM24}
Peter Keevash, Noam Lifshitz, Eoin Long, and Dor Minzer.
\newblock Hypercontractivity for global functions and sharp thresholds.
\newblock {\em J. Amer. Math. Soc.}, 37(1):245--279, 2024.

\bibitem[KLM23]{KLM23}
Nathan Keller, Noam Lifshitz, and Omri Marcus.
\newblock Sharp hypercontractivity for global functions.
\newblock {\em arXiv preprint arXiv:2307.01356}, 2023.

\bibitem[KLM24]{keevash2024largest}
Peter Keevash, Noam Lifshitz, and Dor Minzer.
\newblock On the largest product-free subsets of the alternating groups.
\newblock {\em Inventiones mathematicae}, 237(3):1329--1375, 2024.

\bibitem[KM22]{KM}
Tali Kaufman and Dor Minzer.
\newblock Improved optimal testing results from global hypercontractivity.
\newblock In {\em {FOCS} 2022}, pages 98--109, 2022.

\bibitem[KMS17]{KMS1}
Subhash Khot, Dor Minzer, and Muli Safra.
\newblock On independent sets, 2-to-2 games, and {Grassmann} graphs.
\newblock In {\em {STOC} 2017}, pages 576--589, 2017.

\bibitem[KMS23]{KMS}
Subhash Khot, Dor Minzer, and Muli Safra.
\newblock Pseudorandom sets in grassmann graph have near-perfect expansion.
\newblock {\em Annals of Mathematics}, 198(1):1--92, 2023.

\bibitem[LM19]{LM}
Noam Lifshitz and Dor Minzer.
\newblock Noise sensitivity on the $p$-biased hypercube.
\newblock In {\em {FOCS} 2019}, pages 1205--1226, 2019.

\bibitem[LM23]{LMgp}
Noam Lifshitz and Avichai Marmor.
\newblock Bounds for characters of the symmetric group: A hypercontractive approach.
\newblock {\em arXiv preprint arXiv:2308.08694}, 2023.

\bibitem[MZ23]{MZrmgen}
Dor Minzer and Kai~Zhe Zheng.
\newblock Optimal testing of generalized {R}eed-{M}uller codes in fewer queries.
\newblock In {\em {FOCS} 2023}, pages 206--233. {IEEE}, 2023.

\bibitem[MZ24]{MZoptimal}
Dor Minzer and Kai~Zhe Zheng.
\newblock Near optimal alphabet-soundness tradeoff {PCP}s.
\newblock In {\em {STOC} 2024}, pages 15--23. {ACM}, 2024.

\bibitem[O'D14]{o2014analysis}
Ryan O'Donnell.
\newblock {\em Analysis of boolean functions}.
\newblock Cambridge University Press, 2014.

\bibitem[RM97]{RM97}
R.~Raz and P.~McKenzie.
\newblock Separation of the monotone {NC} hierarchy.
\newblock In {\em Proceedings 38th Annual Symposium on Foundations of Computer Science}, pages 234--243, 1997.

\bibitem[SSSV23a]{saxena2023streaming}
Raghuvansh~R Saxena, Noah Singer, Madhu Sudan, and Santhoshini Velusamy.
\newblock Streaming complexity of {CSP}s with randomly ordered constraints.
\newblock In {\em Proceedings of the 2023 Annual ACM-SIAM Symposium on Discrete Algorithms (SODA)}, pages 4083--4103. SIAM, 2023.

\bibitem[SSSV23b]{saxena2023improved}
Raghuvansh~R Saxena, Noah~G Singer, Madhu Sudan, and Santhoshini Velusamy.
\newblock Improved streaming algorithms for maximum directed cut via smoothed snapshots.
\newblock In {\em 2023 IEEE 64th Annual Symposium on Foundations of Computer Science (FOCS)}, pages 855--870. IEEE, 2023.

\bibitem[SSSV25]{saxena2025streaming}
Raghuvansh~R Saxena, Noah~G Singer, Madhu Sudan, and Santhoshini Velusamy.
\newblock Streaming algorithms via local algorithms for maximum directed cut.
\newblock In {\em Proceedings of the 2025 Annual ACM-SIAM Symposium on Discrete Algorithms (SODA)}, pages 3392--3408. SIAM, 2025.

\bibitem[Sud22]{SudanSurvey}
Madhu Sudan.
\newblock Streaming and sketching complexity of csps: A survey.
\newblock {\em arXiv preprint arXiv:2205.02744}, 2022.

\bibitem[YZ24]{YZ24}
Guangxu Yang and Jiapeng Zhang.
\newblock Communication lower bounds for collision problems via density increment arguments.
\newblock In {\em Proceedings of the 56th Annual ACM Symposium on Theory of Computing}, STOC 2024, page 630–639, New York, NY, USA, 2024. Association for Computing Machinery.

\end{thebibliography}
\appendix
\section{\Cref{thm:main} Implies~\Cref{thm:main_into}}\label{sec:apx_quick_pf}
The proof that~\Cref{thm:main} implies~\Cref{thm:main_into} 
follows a standard reduction from communication protocols to streaming algorithms. Such arguments appear in \cite{KKS15,KK19} in the context of one-pass streaming algorithms, and below we give a straightforward extension to the multi-pass setting for the sake of completeness. 
\begin{proof}[Proof of \Cref{thm:main_into}]
    Suppose we have a $s$-space streaming algorithm $\mathcal{A}$ that, given $k$ passes over a data stream of edges of a graph $G$, outputs a $\left( \frac{1}{2}+\eps\right)$-approximation to the value of the maximum cut of $G$ with probability at least $0.99$. We show that we can construct a communication protocol $\Pi$ that solves $\textsf{DIHP}(n,\alpha ,K)$, where $K = 512/(\alpha\eps^2)$, with advantage $\adv(\Pi)\geq 0.1$ using $k\cdot s\cdot K$ bits of communication, and then applying~\Cref{thm:main} gives the lower bound on $s$.

    Let $\alpha = 10^{-7},K = 512/(\alpha \eps^2)$. Given an instance $(y^{(1)},y^{(2)},\dots,y^{(K)})$ of $\text{DIHP}(n,\alpha,K)$, we construct a data stream of edges as follows: for each $i=1,\ldots,K$, player $i$ constructs a part of the data stream $\sigma^{(i)}$ using $y^{(i)}$. 
    The part $\sigma^{(i)}$ consists of all edges $\{u,v\}\in \supp(y^{(i)})$ such that $y^{(i)}_{uv} =1$. The final data stream $\sigma$ is the concatenation of the $K$ parts constructed above, namely  $\sigma: = \sigma^{(1)}\circ	\sigma^{(2)}\circ\dots\circ\sigma^{(K)}$. We use $G$ to denote the graph that contains the edges in $\sigma$. We use $\mathcal{D}^N_{graph}$ and $\mathcal{D}^Y_{graph}$ to denote the distribution of $G$ when $(y^{(1)},y^{(2)},\dots,y^{(K)})$ follows $\Dno$ and $\Dyes$ respectively. Note that $\mathcal{D}^N_{graph}$ and $\mathcal{D}^Y_{graph}$ may be multigraphs with repeated edges here. The expected number of repeated edges is $O(\frac{1}{\eps^2})$, and edge multiplicities are bounded by $2$ with high probability. 
    We use the following result from \cite{KK19}, stating that $\mathcal{D}^N_{graph}$ and  $\mathcal{D}^Y_{graph}$ have very different max-cut values. 

    \begin{lemma}[Lemma 2.4, \cite{KK19}]\label{lem:separatingYesNo}
        For every $\alpha, \eps \in (n^{-1/10},1),\alpha<1/4$, $K= \frac{512}{\alpha\eps^2 }$, and large enough $n$, there exists $k_0 = k_0(n,\alpha,K)$ such that $G\sim \mathcal{D}_{graph}^N$ has Max-Cut value at most $k_0/(2-\eps)$ and $G\sim \mathcal{D}_{graph}^Y$ has Max-Cut value at least $k_0$ both with probability at least $1-\frac{1}{\sqrt{n}}$. 
    \end{lemma}
    With this lemma, we can then transform $\mathcal{A}$ into a communication protocol $\Pi$ for $\textsf{DIHP}(n,\alpha,K)$ as follows: the first player locally simulates $\mathcal{A}$ on $\sigma^{(1)}$, and sends the memory state to player $2$ after processing all data in $\sigma^{(1)}$; the second player continues the simulation of $\mathcal{A}$ on $\sigma^{(2)}$, and sends the memory state to player $3$ after the simulation; the process continues until the player $K$ after the simulation of the first pass of $\mathcal{A}$ on $\sigma$, it then sends the memory state back to player $1$ again and player $1$ continues the simulation of the second pass of $\mathcal{A}$. After $k$ passes of simulation of $\mathcal{A}$, the protocol $\Pi$ outputs ``yes ''if the output of $\mathcal{A}$ is bigger than $k_0/(2-\eps)$ and outputs ``no'' otherwise. 

    It is easy to check that the protocol $\Pi$ only costs $k\cdot s\cdot K $ bits of communication. 

    For correctness, by Lemma \ref{lem:separatingYesNo} and our assumptions of $\mathcal{A}$, we know that in the yes case, $\Pi$ outputs ``yes'' with probability at least $0.99 - \frac{1}{\sqrt{n}}$ since $(\frac{1}{2}+\eps)\geq \frac{1}{2-\eps}$; in the no case, $\Pi$ outputs ``yes'' with probability at most $0.01 + \frac{1}{\sqrt{n}}$. As a result, the advantage of $\Pi$ is lower bounded by $0.1$ when $n$ is large enough. 
\end{proof}
\end{document}